\documentclass[a4paper,english,12pt,reqno]{amsart}

\usepackage{amsmath, amsthm, amssymb}
\usepackage{a4, babel, graphicx, listings, fancyhdr, verbatim}
\usepackage{mathrsfs} 
\usepackage{anysize} 
\usepackage{natbib}
\usepackage[format = hang]{caption}                              

\usepackage{a4}
\usepackage[latin1]{inputenc}
\usepackage{amsmath, amsthm}
\usepackage{graphicx}
\usepackage{amssymb}
\usepackage{verbatim}
\usepackage{listings}
\usepackage{alltt}
\usepackage{babel}
\usepackage{mathrsfs}
\usepackage{ae}
\usepackage{natbib}

\newtheoremstyle{theorem}
{10pt} 
{10pt} 
{\sl} 
{\parindent} 
{\bf} 
{. } 
{ } 
{} 

\theoremstyle{theorem}

\newtheorem{lemma}{Lemma}

\newtheorem{assumption}{Assumption}
\newtheorem{remark}{Remark}

\def\beq{\begin{eqnarray}}
\def\eeq{\end{eqnarray}}

\def\beqn{\begin{eqnarray*}}  
\def\eeqn{\end{eqnarray*}}

\def\E{{\rm E}}

\def\dd{{\rm d}}
\def\N{{\rm N}}

\def\Pr{P}
\def\pr{{\rm pr}}

\def\quadandquad{\quad {\rm and} \quad}
\def\arr{\rightarrow}
\def\hatt{\widehat}
\def\tilda{\widetilde}
\def\sumin{\sum_{i=1}^n}

\def\eps{\varepsilon}
\def\half{\hbox{$1\over2$}}

\def\quart{\hbox{$1\over4$}}

\def\rootn{\sqrt{n}}
\def\data{{\rm data}}

\def\midd{\,|\,}
\def\tr{{\rm t}}
\def\dell{\partial}

\def\prof{{\rm prof}}
\def\KL{{\rm KL}}

\def\argmin{{\rm argmin}}
\def\argmax{{\rm argmax}}
\def\Bernstein{{Bernshte\u\i n}}
\def\true{{\rm true}}
\def\const{{\rm const.}}

\def\AIC{{\rm AIC}}
\def\BIC{{\rm BIC}}

\def\AJIC{{\rm AJIC}}
\def\BJIC{{\rm BJIC}}
\def\Tr{{\rm Tr}}

\numberwithin{equation}{section} 
\numberwithin{figure}{section}
\numberwithin{table}{section}
\usepackage{hyperref}

\title[JIC for jump regression models]{Estimation, 
   inference and model selection \\ for jump regression models}

\begin{document}


\maketitle

\centerline{\bf Steffen Gr\o nneberg$^1$, Gudmund Hermansen$^{2,3}$ 
   and Nils Lid Hjort$^3$}

\medskip 
\centerline{$^1$BI Norwegian Business School, Oslo,} 
\centerline{$^2$Norwegian Computing Centre, Oslo,} 
\centerline{$^3$Department of Mathematics, University of Oslo} 

\medskip
\centerline{\bf June 2014}

\begin{abstract}
\noindent 
We consider regression models with data of the type
$y_i=m(x_i)+\eps_i$, where the $m(x)$ curve is taken
locally constant, with unknown levels and jump points.
We investigate the large-sample properties of 
the minimum least squares estimators, finding in particular
that jump point parameters and level parameters
are estimated with respectively $n$-rate precision
and $\rootn$-rate precision, where $n$ is sample size.
Bayes solutions are investigated as well and found
to be superior. We then construct jump information criteria,
respectively AJIC and BJIC, for selecting the right
number of jump points from data. This is done by 
following the line of arguments that lead to the
Akaike and Bayesian information criteria AIC and BIC,
but which here lead to different formulae due to 
the different type of large-sample approximations involved.
\end{abstract}



\section{Introduction}
\label{section:intro}

Consider regression data $(x_i,y_i)$, with $x_i$ and $y_i$
one-dimensional. There is a long list of regression models
of the type 
\beq
\label{eq:intro1}
y_i=m(x_i,\theta)+\eps_i
   \quad {\rm for\ }i=1,\ldots,n, 
\eeq 
where the regression function $m(x,\theta)$ is 
either linear in certain basis functions, 
like for polynomial regression, or smooth in $x$,
and where the $\eps_i$ are zero-mean i.i.d.~errors
with standard deviation level $\sigma$. Here 
least squares estimators for the $\theta$ parameter 
can be worked with and leads to the familiar type of 
$\rootn$ convergence to normality, and in its turn
to inference methods, along with strategies 
for model selection of the AIC and BIC variety. 
This is the topic of Section \ref{section:smooth} below, 
encompassing also non-linear but smooth regression models.
Though that section partly contains ostensibly 
well-known material on least squares estimation in 
smooth regression models we emphasise 
that the developments there pertaining to 
a model robust AIC strategy and to refined versions
of the BIC are new. 

The main focus of this article is however estimation,
inference and model selection methods for step function 
regression models, involving discontinuity points
and levels for $m(x,\theta)$. Here the methodology
pans out differently from the smooth case, and 
we shall in particular see that there is 
$n$-convergence rather than $\rootn$-convergence
for estimators of break points and that limit 
distributions are not normal. 

For convenience and without essential loss of generality 
we take the $x$ range to be the unit interval, and 
study the model where 
\beq 
\label{eq:intro2}
m(x,\theta)=a_j \quad {\rm for\ }\gamma_{j-1}\le x<\gamma_j, 
\eeq 
for windows $[\gamma_{j-1},\gamma_j]$, with $j=1,\ldots,d$
and $\gamma_0=0,\gamma_d=1$. The unknown parameters to
estimate from data are the $d-1$ break point positions
and the $d$ levels, along with the spread parameter $\sigma$.
For an illustration, see Figure \ref{figure:figureA},
with the true $m(x)$ function of the above type
having five windows. Supposing the number of windows
being known to the statistician, the task there 
is to estimate and carry out inference for 
the five level parameters, the four break point locations,
and $\sigma$, based on such a data set. Often 
one would not know the number of windows a priori,
however, in which case inferring the right number
of break points is also part of the task. 
Our paper provides methods for handling these
questions. 


The simplest non-trivial version of this model
has two windows, with say
\beq 
\label{eq:intro3}
m(x,\theta)= \left\{\begin{matrix}a & \text{if } x\le \gamma, \\
                     b & \text{if } x>\gamma, \end{matrix} \right. 
\eeq 
and with both levels $a,b$ and break point position $\gamma$
unknown parameters. We derive the basic large-sample properties
of the least squares method in Section \ref{section:twowindowsA}; 
specifically, it is shown for the least squares estimator
$\hatt\gamma$ that $n(\hatt\gamma-\gamma_0)$ has 
a limit distribution, associated with a certain
two-sided compound Poisson process. This means a faster 
rate of convergence than in traditional parametric models. 
The situation where there is a break point, 
but where the level of $m(x)$
is not necessarily entirely flat throughout the two windows,
is also considered, in Section \ref{section:twowindowsB}.
Understanding that situation is vital not only for
constructing model robust inference methods, but 
for building model selection strategies. 
Methods and results are then generalised in 
Section \ref{section:multiwindows} 
to the case of (\ref{eq:intro2}),
i.e.~with $d$ windows and $2d$ unknown parameters.
There is again the faster $n$-rate convergence 
for estimators of the $d-1$ break points. 

\begin{figure}
\centering
\includegraphics[width = \textwidth]{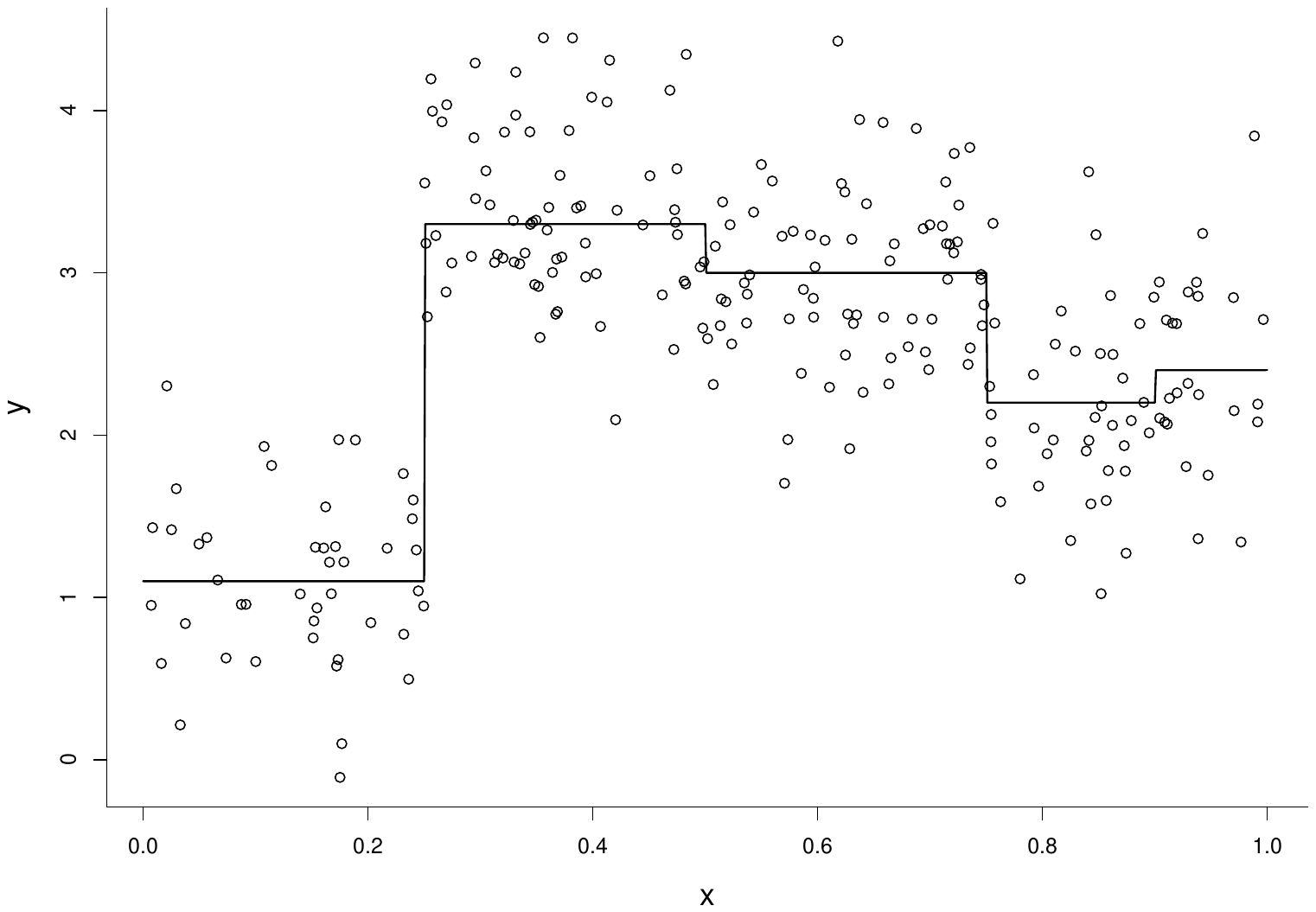}
\caption{
Illustration of the (\ref{eq:intro2}) model, with 
five windows and $n=250$ data points. The statistical 
task is to estimate the discontinuity points and 
the levels from the data, and also to infer 
the right number of windows in case this is not 
known a priori.}
\label{figure:figureA}
\end{figure} 

In some types of applications the number of windows
would be known to the analyst a priori, say $d$, 
hence limiting the statistical inference task to 
the $d-1$ break points and $d$ levels. 
Frequently the number of windows would not be known
a priori, as pointed to above.
The difficulty clearly depends on the 
sizes of windows and degree of separation between
their levels; in Figure \ref{figure:figureA} 
the small difference between the levels of 
windows four and five would be hard to identify,
unless the sample size is large. In 
Sections \ref{section:ajic} and \ref{section:bjic}
we develop jump information criteria (JIC),
respectively the AJIC and BJIC, for selecting the 
right number of windows based on data. These
are constructed by following the basic line
of arguments which for regular models lead to 
respectively the AIC and the BIC formulae,
cf.~Section \ref{section:smooth}, but which now
are in need of modification and sophistication,
in view of the different type of asymptotics at work. 

In Section \ref{section:bayes} we study the Bayesian 
approach to estimation problems in models of the 
(\ref{eq:intro2}) form, and find that the Bayes 
estimators $\hatt\gamma_{j,B}$ have the same 
fast convergence rate as for the least squares method, 
but that the limit distributions of 
$n(\hatt\gamma_{j,B}-\gamma_{0,j})$ are different. 
We also show that the Bayes estimators have 
smaller asymptotic mean squared errors. 
Then in Section \ref{section:algorithms} various 
computational aspects are addressed. We are in particular 
exhibiting a new version of previously suggested procedures
which uses dynamic programming to rapidly compute 
both least squares and Bayes estimates; the algorithms
work fast even for high $n$ and with many break points.
Illustrations of our methods and findings are given
in Section \ref{section:illustrations}, and our paper 
ends with a list of concluding remarks and suggestions 
for future work in Section \ref{section:concludingremarks}.

The themes worked with in this paper have 
quite a vast literature, and to complicate matters, 
several asymptotic perspectives exist.
 Some book length treatments 
on the subject for a single change point are 
\citet{BassevilleNikiforov93}, 
\citet{BrodskyDarkhovsky93}, 
\citet{Carlsteinetal94} and 
\citet{CsorgoHorvath97}. 
Modern treatments of a single change point for jump regression
models include Section 14.5.1 of \citet{Kosorok08} and the
application in \citet{SeijoSen11b}.

The problem of multiple change points is less studied in
the pure statistics literature, but has been thoroughly 
investigated in the econometric literature.
The asymptotic perspective in the econometric literature 
is that the sizes of the jumps shrink as the number of 
observations increases. In the one-jump case, one assumes 
that for some $0 < \alpha < 1/2$ the jump size $\lambda_n$
is such that $\lambda_n \rightarrow 0$ as 
$n \rightarrow \infty$, but that $n^{1/2-\alpha} \lambda_n
\rightarrow 0$ as $n \rightarrow \infty$.
In contrast, we will follow \citet{Kosorok08, SeijoSen11b} and
consider the covariate placements as random, 
as formalised in the upcoming Assumption \ref{assumption:xxiid}
-- but let the jump sizes be fixed. See also Remark C in 
Section \ref{section:concludingremarks}.

Two central papers on frequentistic estimation from the
econometric perspective are 
\citet{BaiPerron98} and \citet{BaiPerron03}.
In \citet{BaiPerron98}, the limit distributions in 
the presence of a prescribed number of jumps is derived using 
the shrinking jump perspective for its asymptotic approximations.
\citet{BaiPerron03}
surveys the practical aspects 
of frequentistic estimation of regression models with jumps.

A large-scale simulation study of the small-sample approximation 
of the limit law of the
estimators in the presence of multiple jumps is found in 
\citet{BaiPerron06}, where shrinking jump sizes are assumed.
They conclude that `the coverage rates 
for the break dates are adequate unless the break is either 
too small (so small as not to be detected by the tests) or 
too big'. 
Such a large-scale
simulation has to our knowledge not been conducted for
the case where the jump sizes are fixed.
It seems reasonable that some of the qualitative conclusions 
from \citet{BaiPerron06} are valid for our case as well.


We note a finite-sample approach to the Gaussian case is 
\citet{Lebarbier05}. 

On the Bayesian side, there is again a large literature 
available. 
\citet{Fearnhead06} and \citet{Chib98} 
both develop computationally feasible estimation methods. 
There is even an $\mathsf{R}$ package available, 
described in \citet{ErdmanEmerson07}, 
which also gives many useful references. 

We note that there are several distinct models that are technically
similar to the jump regression models of the present paper.
Threshold regression models is a prominent example, see
\cite{Chan93, Yu12, LiLing12}. 

The case of continuous process (signal) with discontinuities 
has been worked with, both from Bayesian and frequentist 
perspectives, in 
\citet[p.~329--338]{IbragimovKhasminski81}. 
They work with stochastic differential equations of the form
$\dd X(t) = S(t-\theta)\,\dd t + \dd B(t)$
for $t\in[0, n]$, 
where $B(t)$ is a Brownian motion process, $S(u)$ 
is a periodic deterministic function of period $1$ 
with discontinuities of the first kind, and $n$ is an integer. 


There are relatively speaking far fewer papers in the
literature seriously addressing the model selection
issues, the task of finding the right number of 
windows in the data. Some care is needed regarding
the interpretation and assessments of such methods, 
as the task of pinpointing and counting genuine 
discontinuities in the underlying signal, which is 
our perspective in this article, is different from
approximating a smooth function with a locally constant
curve, as in e.g.~\citet{BanerjeeMcKeague07} 
and \citet{HermansenHjort14a}. For general perspectives
on modelling with discontinuities, see the discussion of 
\citet{FrigessiHjort02} and the articles collected
in the associated special issue of the 
Journal of Nonparametric Statistics.
Some papers discussing model selection issues 
for jump regressions are \citet{Ninomiya05, Munketal14}. 

A recent work, also dealing with the selection of jumps is, 
\citet{Zouetal14}. They use a nonparametric maximum likelihood
approach, while we will exclusively use 
models parametrised by a finite-dimensional and fixed set.
However, their model selection technique focuses on
providing a procedure that selects the correct number of
jumps in a consistent manner as $n \rightarrow \infty$. 
Our model selection ideas are instead based on extending the
motivation behind the AIC and BIC formulae, as set up
and developed in
\citet[Chapters 2 \& 3]{ClaeskensHjort08}. 
Other papers that extend the underlying motivation of the AIC
and BIC to non-classical settings include 
\citet{HermansenHjort14a, HermansenHjort14b, GroennebergHjort14}. 

Applications of change-point models of the 
(\ref{eq:intro2}) variety abound, and include the
fields of biology, geology, engineering with e.g.~speech 
recognition, medicine, imaging, marine biology and oceanography, 
etc. A broad overview is offered in \citet{Munketal14},
including references. \citet{Kjesbuetal14} have 
extended a classic time series from \citet{Hjort1914},
pertaining to the liver quality of fish, from the
original 1880--1912 range to 1859--2012, and models used
to understand, analyse, interpret and forming predictions 
from this and similar time series for marine sciences 
include versions of (\ref{eq:intro2}). 
In this connection we also mention the so-called 
bent cable regression models of \citet{Chiuetal06},
where the mean curve is taken to be continuous
but allowed to have discontinuities in its derivative; 
\citet{Kjesbuetal14} discuss also such models for
the Hjort time series.  

\section{The case of smooth regression models}
\label{section:smooth}

It is useful to study the more traditional case 
of smooth regression curves first, before 
tendind to the jump regression models. 
Here the theory indeed works more smoothly, regarding 
both estimation, inference and model selection, 
leading in particular to formulae for AIC and BIC 
close to the familiar ones,
see e.g.~\citet{ClaeskensHjort08} for these. 
We shall see that several constructions and properties
rely on analysis of two special functions 
worked with below, the $A_n(s,t)$ process and the $B_n(s,t)$ 
function, associated respectively with 
the log-likelihood function and a Kullback--Leibler
distance function. When we turn to jump regression 
models in later sections it becomes necessary 
to understand and analyse more technically demanding 
analogues of $A_n$ and $B_n$ defined for these 
non-smooth models.

Since some of our limit distribution results 
and associated large-sample approximations depend
not only on the model formulation and the $m(x)$
function, but also on how the $x_i$ are distributed,
we put up an explicit assumption, as follows.

\begin{assumption}
\label{assumption:xxiid}
{{\rm 
The $x_i$ are generated independently from a covariate
distribution with density $f(x)$ positive on the 
unit interval. Hence ergodic averages 
$\bar h_n=n^{-1}\sumin h(x_i)=\int h(x)\,\dd F_n(x)$
converge to limits $\int h(x)f(x)\,\dd x$, 
for each bounded $h$. Here $F_n$ is the empirical 
distribution of the observed $x_1,\ldots,x_n$,
converging to $F$, the cumulative of $f$. }}
\end{assumption}

\subsection{The model and its least false parameters}

Consider then a setup in which 
\beq
\label{eq:realmodel}
y_i=m(x_i)+\eps_i \quad {\rm for\ }i=1,\dots,n, 
\eeq 
where $m(x)=m_\true(x)$ is the true regression curve 
$\E\,(Y\midd x)$ and where the $\eps_i$ are 
i.i.d.~$\N(0,\sigma_\true^2)$. 
A parametric model $m(x,\theta)$ is then suggested,
in this section taken to be smooth, with two continuous
derivatives in both $x$ and $\theta$.
Prototype cases would be the polynomial
$\theta_0+\theta_1x+\cdots+\theta_px^p$, or an
expansion in other basis functions, but also 
nonlinear models like 
e.g.~$\theta_0+\theta_1\{1-\exp(-\theta_2 x_i)\}$ 
fit the general framework. We shall first
clarify what the ML estimators are aiming at, when
we avoid postulating that the parametric model 
is correct, i.e.~$m(x)$ is not assumed to be $m(x,\theta)$
for any $\theta$. 

The log-likelihood function takes the form 
$$\ell_n(\theta,\sigma)=-n\log\sigma
   -\half(1/\sigma^2)Q_n(\theta) $$
(ignoring the immaterial constant $-\half n\log(2\pi)$ 
here and in what follows), 
where $Q_n(\theta)=\sumin \{y_i-m(x_i,\theta)\}^2$. 
We see that the ML estimator $\hatt\theta$ is 
the argmin of $Q_n$ and that the ML estimator of
$\sigma$ is given by 
$\hatt\sigma^2=Q_n(\hatt\theta)/n=Q_{n,\min}/n$. 
We have 
$$\E\,Q_n(\theta)/n=\sigma_\true^2+n^{-1}\sumin \{m(x_i)-m(x_i,\theta)\}^2, $$
and we define the least false regression parameter $\theta_{0,n}$ 
to be the parameter value minimising this expression, 
i.e.~minimising $\int \{m(x)-m(x,\theta)\}^2\,\dd F_n(x)$.
We similarly define the least false spread parameter 
$\sigma_{0,n}$ via 
\beq
\label{eq:sigmasigma}
\sigma_{0,n}^2=\sigma_\true^2
   +\int\{m(x)-m(x,\theta_{0,n})\}^2\,\dd F_n(x). 
\eeq 
These least false parameters are those at which the 
ML estimators $\hatt\theta$ and $\hatt\sigma$ are aiming,
and correspond to the best possible approximation
afforded by the parametric model to the real 
data generating mechanism (\ref{eq:realmodel}). 
When $n$ increases, $\theta_{0,n}\arr\theta_0$,
the minimiser of $\int \{m(x)-m(x,\theta)\}^2 f(x)\,\dd x$,
and $\sigma_{0,n}\arr\sigma_0$, given by 
$\sigma_0^2=\sigma_\true^2+\int\{m(x)-m(x,\theta_0)\}^2 f(x)\,\dd x$.
Note that a good model succeeds in $m(x,\theta_0)$
coming close to the real $m(x)$, hence also 
leading to a $\sigma_0$ being small and coming 
close to $\sigma_\true$ of (\ref{eq:realmodel}),
whereas a coarser model typically does not 
come close to $m(x)$ and implies a larger 
least false parameter $\sigma_0$. 

\subsection{The log-likelihood function and the $A_n$ process}

Consider now the random function
\beq
\label{eq:An}
A_n(t,u)=\ell_n(\theta_{0,n}+t/\rootn,\sigma_{0,n}+u/\rootn)
   -\ell_n(\theta_{0,n},\sigma_{0,n}). 
\eeq
It will be seen that some of the more important 
results concerning the distributions of the ML
estimators follow from large-sample properties 
of the $A_n$ function, along with natural model 
selection criteria.  
To work with $A_n$ and its distributional limit, we first need 
$$W_n=\rootn(\tilda\sigma^2/\sigma_{0,n}^2-1)\arr_d W\sim\N(0,2), $$ 
where $\tilda\sigma^2=n^{-1}\sumin\{y_i-m(x_i,\theta_{0,n})\}^2$.
Secondly, 
\beqn
H_n(t)&=&Q_n(\theta_{0,n}+t/\rootn)-Q_n(\theta_{0,n}) \\
   &=&\sumin \bigl[\{y_i-m(x_i,\theta_{0,n})-d_i\}^2
   -\{y_i-m(x_i,\theta_{0,n})\}^2\bigr] \\
   &=&\sumin \bigl[-2\{y_i-m(x_i,\theta_{0,n})\}d_i+d_i^2\bigr],
\eeqn
in which 
\beq
\label{eq:ditaylor}
\hskip-0.5truecm
d_i=m(x_i,\theta_{0,n}+t/\rootn)-m(x_i,\theta_{0,n})
   \doteq m^*(x_i,\theta_0)^\tr t/\rootn
   +\half t^\tr m^{**}(x_i,\theta_0)t/n, 
\eeq
writing $m^*(x,\theta)$ and $m^{**}(x,\theta)$
for respectively the $p$-vector and $p\times p$-matrix
of first and second order partial derivatives 
with respect to $\theta$, and where $p$ is the 
dimension of $\theta$. 
We find from this and some additional efforts that 
$H_n(t)=-2t^\tr U_n+t^\tr(\Sigma_n+\Omega_n)t+o_\pr(1)$.
Here 
$$\Sigma_n=n^{-1}\sumin m^*(x_i,\theta_{0,n})m^*(x_i,\theta_{0,n})
   \arr\Sigma=\int m^*(x,\theta_0)m^*(x,\theta_0)f(x)\,\dd x $$
and 
\beqn
\Omega_n&=&-n^{-1}\sumin \{m(x_i)-m(x_i,\theta_{0,n})\}
   m^{**}(x_i,\theta_{0,n}) \\
   & &\arr\Omega=\int \{m(x,\theta_0)-m(x)\}m^{**}(x,\theta_0)f(x)\,\dd x, 
\eeqn
whereas 
$$U_n=n^{-1/2}\sumin\{y_i-m(x_i,\theta_{0,n})\}m^*(x_i,\theta_{0,n})
   =n^{-1/2}\sumin \eps_im^*(x_i,\theta_{0,n}). $$
We use here the fact that 
\beq
\label{eq:productiszero}
\sumin\{m(x_i)-m(x_i,\theta_{0,n})\}m^*(x_i,\theta_{0,n})=0, 
\eeq 
which follows from the definition of the least false parameter 
$\theta_{0,n}$. This then leads to 
$U_n\arr_d U\sim\N(0,\sigma_\true^2\Sigma)$. 
There is also joint convergence $(W_n,U_n)\arr_d(W,U)$,
with $W$ and $U$ found to be independent,
again thanks to (\ref{eq:productiszero}). It follows that 
$$H_n(t)\arr_d H(t)=-2t^\tr U+t^\tr(\Sigma+\Omega)t. $$
We note that $\Omega$ is zero or small 
if the parametric model is correct or nearly correct
or if $m(x,\theta)$ is linear or nearly linear in $\theta$.
For linear regression models, for example, $\Omega_n$
and hence $\Omega$ is zero.

Combining these ingredients one may now establish 
the limit process for $A_n$ of (\ref{eq:An}). We have 
\beqn
A_n(t,u)&=&-n\log(\sigma_{0,n}+u/\rootn)
   -\half{1\over \sigma_{0,n}^2}{1\over \{1+(u/\sigma_{0,n})/\rootn\}^2}
   \{Q_n(\theta_{0,n})+H_n(t)\} \\
& &   +n\log\sigma_{0,n}+\half{1\over \sigma_{0,n}^2}Q_n(\theta_{0,n}) \\
&=&-n\log\{1+(u/\sigma_{0,n})/\rootn\}
   -\half{1\over \sigma_{0,n}^2}
   \Bigl[{1\over \{1+(u/\sigma_{0,n})/\rootn\}^2}-1\Bigr]
   Q_n(\theta_{0,n}) \\
& &  -\half{1\over \sigma_{0,n}^2}
   {1\over (1+(u/\sigma_{0,n})/\rootn)^2}H_n(t). 
\eeqn
This leads with some further details and efforts to 
\beq
\label{eq:AntoA}
\begin{array}{rcl} 
A_n(t,u)&=&{u\over \sigma_0}W_n-{u^2\over \sigma_0^2}
   -\half{1\over \sigma_0^2}H_n(t)+o_\pr(1) \displaystyle \\ 
&\arr_d& A(t,u)
   ={u\over \sigma_0}W-{u^2\over \sigma_0^2}
   +{1\over \sigma_0^2}\{t^\tr U-\half t^\tr(\Sigma+\Omega)t\}. 
   \displaystyle  
\end{array}
\eeq 
The limit function depends on $(W,U)$ only, 
and, after securing tightness -- which we 
will not do here -- 
the convergence in distribution here 
takes place not merely pointwise, 
i.e.~for each $(t,u)$, but inside each function space 
$D([-r,r]^{p+1})$, for $r$ positive, equipped with
the Skorokhod topology for right-continuous 
functions with left-hand limits; 
see \citet{Billingsley68} for mathematical details
concerning convergence of probability measures
in such function spaces. 
 The point, in our connection,
is that $A_n\arr_d A$ automatically implies 
$h(A_n)\arr_d h(A)$ for each $h$ continuous 
on the support of $A$. Various useful 
corollaries hence flow from (\ref{eq:AntoA}),
as we now demonstrate. 
 

(i) Using the argmax continuous mapping theorem as described
e.g.~in \citet[Chapter 3.2.1]{vdVaartWellner96}, we 
get that if $\rootn(\hatt\theta-\theta_{0,n})$ is 
bounded in probability then $\argmax(A_n)\arr_d\argmax(A)$: 
\beq
\label{eq:thatuhat}
\begin{array}{rcl}
\hatt t_n=\rootn(\hatt\theta-\theta_{0,n})
   &\arr_d&
   \hatt t=(\Sigma+\Omega)^{-1}U, \displaystyle \\
\hatt u_n=\rootn(\hatt\sigma-\sigma_{0,n})
   &\arr_d&
   \hatt u=\half\sigma_0 W. \displaystyle
\end{array}
\eeq
Also, these maximisers of the random limit function $A$ 
have distributions 
$$\hatt t\sim\N_p(0,\sigma_\true^2(\Sigma+\Omega)^{-1}
   \Sigma(\Sigma+\Omega)^{-1})
   \quadandquad \hatt u\sim\N(0,\half\sigma_0^2). $$ 
Note that we do not provide sufficient conditions for
$\rootn(\hatt\theta-\theta_{0,n}) = O_P(1)$, as this
is a standard problem investigated for example 
in \citet[Chapters 7 \& 9]{van2000empirical}.

(ii) Consider 
$$A_n^*(t,u)=A_n(\hatt t_n+t,\hatt u_n+u)-A_n(\hatt t_n,\hatt u_n)
   =\ell_n(\hatt\theta+t/\rootn,\hatt\sigma+u/\rootn)
   -\ell_n(\hatt\theta,\hatt\sigma). $$
We have 
\beq
\label{eq:AnstartoAstar}
A_n^*(t,u)\arr_d A^*(t,u)=A(\hatt t+t,\hatt u+u)-A(\hatt t,\hatt u)
   =-u^2-\half {1\over \sigma_0^2}t^\tr(\Sigma+\Omega)t. 
\eeq 
This relates to the fact that there is joint convergence
in distribution $(A_n,\hatt t_n)\arr_d (A,\hatt t)$, and
that the composite function taking $(A,t)$ to $A(t)$
is continuous. 

(iii) We also have 
$$A_n(\hatt t_n,\hatt u_n)=\ell_n(\hatt\theta,\hatt\sigma)
   -\ell_n(\theta_{0,n},\sigma_{0,n})
   \arr_d A(\hatt t,\hatt u)
   =\quart W^2+\half(1/\sigma_0^2)U^\tr(\Sigma+\Omega)^{-1}U. $$

\subsection{The Kullback--Leibler distance and the $B_n$ function}

The AIC type model selection strategy, worked with
below, is intimately connected to the Kullback--Leibler
distance from the true data generating mechanism
$g(y\midd x)$, here associated with (\ref{eq:realmodel}),
and the estimated parametric model, say 
$f(y\midd x,\hatt\theta,\hatt\sigma)$.
This distance, when averaged over the cases 
$x_1,\ldots,x_n$, is 
\beqn
\KL_n=\int\int g(y\midd x)\log{g(y\midd x)\over 
   f(y\midd x,\hatt\theta,\hatt\sigma)}
   \,\dd y\,\dd F_n(x)
   =\const-n^{-1}k_n(\hatt\theta,\hatt\sigma),
\eeqn 
in which 
\beqn
k_n(\theta,\sigma)
   &=&\sumin\int g_\true(y\midd x_i)
   \log f(y\midd x_i,\theta,\sigma)\,\dd y \\
   &=&-n\log\sigma-\half{1\over \sigma^2}
   \sumin [\{m(x_i)-m(x_i,\theta)\}^2+\sigma_\true^2]. 
\eeqn 
The best approximation or least false parameters
are set up such that they maximise this function,
hence with partial derivatives being zero at 
$(\theta_{0,n},\sigma_{0,n})$. This leads to 
\beqn
B_n(t,u)=k_n(\theta_{0,n}+t/\rootn,\sigma_{0,n}+u/\rootn)
   -k_n(\theta_{0,n},\sigma_{0,n})
   =-\half v^\tr J_nv+o(1), 
\eeqn
writing $v=(t,u)^\tr$, and where 
$$J_n=-n^{-1}{\dell^2 k_n(\theta_{0,n},\sigma_{0,n})
   \over \dell\alpha\dell\alpha^\tr}, $$
writing $\alpha=(\theta,\sigma)$ for the full parameter 
vector. Under the conditions above, we have 
\beqn
J_n\arr J={1\over \sigma_0^2}
\begin{pmatrix}
\Sigma+\Omega, & 0 \\
0, & 2
\end{pmatrix}, 
\eeqn
and $B_n(t,u)\arr B(t,u)=-\half v^\tr Jv$.

\subsection{The model robust AIC}

The model selection strategy underlying the original 
AIC, see \citet[Ch.~3]{ClaeskensHjort08}, is 
in essence to estimate the Kullback--Leibler distance 
from truth to each estimated model and then choosing
the model with smallest such estimated distance.
In view of the above this entails estimating  
$$\hatt k_n=k_n(\hatt\theta,\hatt\sigma)
   =-n\log\hatt\sigma-\half{1\over \hatt\sigma^2}
   \sumin [\{m(x_i)-m(x_i,\hatt\theta)\}^2+\sigma_\true^2] $$
from data, for each candidate model, and then choose
the model with the highest such estimate. 
To do this we start with 
$$\ell_{n,\max}=\ell_n(\hatt\theta,\hatt\sigma)
   =-n\log\hatt\sigma-\half n $$
and investigate to what extent this tends to overshoot
its target $k_n(\hatt\theta,\hatt\sigma)$. 

Write 
\beqn
\ell_{n,\max}
   &=&\ell_n(\theta_{0,n},\sigma_{0,n})
   +\{\ell_n(\hatt\theta,\hatt\sigma)-\ell_n(\theta_{0,n},\sigma_{0,n})\}
   =\ell_n(\theta_{0,n},\sigma_{0,n})+E_{n,1}, \\ 
\hatt k_n
   &=&k_n(\theta_{0,n},\sigma_{0,n})
   +\{k_n(\hatt\theta,\hatt\sigma)-k_n(\theta_{0,n},\sigma_{0,n})\}
   =k_n(\theta_{0,n},\sigma_{0,n})+E_{n,2}. 
\eeqn
The point here is that 
$\E\,\ell_n(\theta_{0,n},\sigma_{0,n})=k_n(\theta_{0,n},\sigma_{0,n})$,
i.e.~that part of $\ell_{n,\max}$ is unbiased for the corresponding
part of $\hatt k_n$, whereas $E_{n,1}$ is positive and 
$E_{n,2}$ is negative, contributing both to a certain 
positive bias for $\ell_{n,\max}$. We first find 
$$E_{n,1}=A_n(\hatt t_n,\hatt u_n)\arr_d 
   E_1=A(\hatt t,\hatt u)
   =\quart W^2+\half(1/\sigma_0^2)U^\tr(\Sigma+\Omega)^{-1}U, $$
with $(\hatt t_n,\hatt u_n)$ and $(\hatt t,\hatt u)$
as in (\ref{eq:thatuhat}). Furthermore, 
$$E_{n,2}=B_n(\hatt t_n,\hatt u_n)\arr_d
   E_2=B(\hatt t,\hatt u), $$
and one finds in fact that 
$(E_{n,1},E_{n,2})\arr_d (E_1,E_2)$ jointly, with $E_2=-E_1$. So 
$$E_{n,1}-E_{n,2}\arr_d E_1-E_2
   =\half W^2+(1/\sigma_0^2)U^\tr(\Sigma+\Omega)^{-1}U, $$
which has the following mean value, 
the overshooting bias of $\ell_{n,\max}$ for estimating
$k_n(\hatt\theta,\hatt\sigma)$: 
$$b=\E\,(E_1-E_2)=1+(\sigma_\true^2/\sigma_0^2)
   \Tr\{(\Sigma+\Omega)^{-1}\Sigma\}. $$
This leads to 
\beq
\label{eq:aicstarsmooth}
\AIC^*=2\ell_{n,\max}-2\hatt b
   =2\ell_{n,\max}-2\bigl[1+(\hatt\sigma^2/\hatt\sigma_0^2)
   \Tr\{(\hatt\Sigma+\hatt\Omega)^{-1}\hatt\Sigma\}\bigr],
\eeq
where $\hatt\sigma^2$ is an estimate of $\sigma_\true^2$. 
We call this the model robust AIC. 

We note that the $\AIC^*$ developed here is not identical
to the traditional AIC, which merely counts estimated 
parameters and here would be equal to 
$\AIC=2\ell_{n,\max}-2(1+p)$. Our version 
(\ref{eq:aicstarsmooth}) is model robust in that we 
are not assuming the parametric model to be correct 
when we assess and estimate the Kullback--Leibler
deviances. For the case of comparing and assessing 
linear regressions models the above simplifies,
in that $\Omega=0$, leading to 
$$\AIC^*=2\ell_{n,\max}-2(1+p\hatt\sigma^2/\hatt\sigma_0^2). $$
Again, $\hatt\sigma_0$ stems from the candidate 
model being assessed, whereas $\hatt\sigma$ 
estimates the underlying $\sigma_\true$. We might use
$\hatt\sigma$ from the biggest model under consideration,
or take $n^{-1}\sumin\{y_i-\tilda m(x_i)\}^2$ with
a nonparametric smoother for $\tilda m(x)$. 

We also record that our $\AIC^*$ is
equivalent to choosing the candidate model with the 
smallest value of 
$$n\log\hatt\sigma_0+(\hatt\sigma^2/\hatt\sigma_0^2)
   \Tr\{(\hatt\Sigma+\hatt\Omega)^{-1}\hatt\Sigma\}. $$
This is also related to Mallows's $C_p$, but in
fact a more robust version, and also encompassing
the possibility of comparing linear with nonlinear
regressions. 

\subsection{The BIC}

Suppose a Bayesian prior $\pi(\theta,\sigma)$ is 
put up for the parameters of the model. We may then 
use (\ref{eq:AnstartoAstar}) to approximate the 
marginal likelihood of the data, with implications
for posterior model computations, as we shall see. 
We find 
\beqn
\lambda_n&=&\int\pi(\theta,\sigma)L_n(\theta,\sigma)
   \,\dd\theta\,\dd\sigma \\
   &=&\exp(\ell_{n,\max})\int\exp\{\ell_n(\theta,\sigma)
   -\ell_n(\hatt\theta,\hatt\sigma)\}
   \pi(\theta,\sigma)\,\dd\theta\,\dd\sigma
\eeqn 
Set now $\theta=\hatt\theta+t/\rootn$ and 
$\sigma=\hatt\sigma+u/\rootn$, which is seen to 
lead to involving $A_n^*(t,u)$ above, and 
\beqn
\lambda_n&=&\exp(\ell_{n,\max})
   \int\exp\{A_n^*(t,u)\}
   \pi(\hatt\theta+t/\rootn,\hatt\sigma+u/\rootn)
   \,\dd t\,\dd u\,(1/\rootn)^{p+1} 
\eeqn 
This leads to the approximation 
$$\log\lambda_n\doteq\ell_{n,\max} -\half(p+1)\log n
   +\log\int\exp\{A^*(t,u)\}\,\dd t\,\dd u
   +\log\pi(\hatt\theta,\hatt\sigma). $$
The integral here may be evaluated exactly, 
in view of (\ref{eq:AnstartoAstar}), giving
$$\int\exp\{A^*(t,u)\}\,\dd t\,\dd u
   =(2\pi)^{(p+1)/2}2^{-1/2}|(\Sigma+\Omega)/\sigma_0^2|^{-1/2} $$
and its turn
\beq
\label{eq:loglambda}
\begin{array}{rcl}
\log\lambda_n&\doteq&\ell_{n,\max}-\half(p+1)\log n \\
& & \qquad 
   +\half(p+1)\log(2\pi)+p\log\sigma_0
   -\half\log|\Sigma+\Omega|
   +\log\pi(\hatt\theta,\hatt\sigma).
\end{array}
\eeq 

Assume now that a list of candidate models 
$\mathcal{M}_1,\ldots,\mathcal{M}_k$ are under 
consideration, along with
priors $\pi(\theta_j,\sigma_j)$ for the parameters 
$(\theta_j,\sigma_j)$ involved in model $\mathcal{M}_j$,
and finally with prior probabilities $p(\mathcal{M}_j)$ 
to make the Bayesian multi-model description complete.
The posterior model probabilities are then 
$p(\mathcal{M}_j\midd\data)\propto 
p(\mathcal{M}_j)\lambda_{n,j}$, with 
$\lambda_{n,j}$ the marginal likelihood computed
under $\mathcal{M}_j$, defined and computed as above. 
We then have 
$$p(\mathcal{M}_j\midd\data)\doteq
  {p(\mathcal{M}_j)\exp\{\ell_{n,j,\max}-\half(p_j+1)\log n+D_{n,j}\}
  \over 
  \sum_{j'=1}^k
  p(\mathcal{M}_{j'})\exp\{\ell_{n,j',\max}-\half(p_{j'}+1)\log n+D_{n,j'}\}}, $$
with $p_j$ the dimension of $\theta_j$ and 
$D_{n,j}$ the relevant extra terms, from the 
approximation developed above. Since these extra
terms are of lower order, in fact of order $O(p)$
compared to order $O_\pr(n)$ and $O(p\log n)$ 
for the two leading terms, the approximation 
$$p(M_j\midd\data)\propto p(\mathcal{M}_j) 
\exp(\half\,\BIC_j) $$
is a practical one, with 
\beq
\label{eq:bicforsmooth}
\BIC_j=2\ell_{n,j,\max}-(p_j+1)\log n. 
\eeq
This is, in essence, the rationale underlying 
the BIC model selection criterion; 
cf.~\citet[Ch.~4]{ClaeskensHjort08}. 
One may of course attempt to use further
aspects of the above development when
approximating $p(\mathcal{M}_j\midd\data)$, 
e.g.~using $(p_j+1)\log\{n/(2\pi)\}$ instead of 
$(p_j+1)\log n$ as suggested by (\ref{eq:loglambda}), 
but this would also need an assessment of 
the actual size of the 
$\log\pi(\hatt\theta,\hatt\sigma)$ term. 
The common practice stemming from this methodology 
is to sidestep these lower-level issues, and also
to take all prior model probabilities $p(\mathcal{M}_j)$
equal; then selecting the most likely model
given the data is equivalent to choosing the
candidate model with highest $\BIC$ score. 

\section{The two-window case: a single discontinuity}
\label{section:twowindowsA}

In this section we focus on the case of a single breakpoint 
parameter $\gamma$ and denote the corresponding level parameters 
$a$ and $b$, to the left and right of $\gamma$, respectively.
The model hence takes the form 
\beq 
\label{eq:twowindowmodel}
y_i=m(x_i,\theta)+\eps_i
  = \left\{\begin{matrix}
   a+\eps_i & \text{if } x\le \gamma, \\
   b+\eps_i & \text{if } x>\gamma, \end{matrix} \right.
\eeq 
writing $\theta=(\gamma,a,b)$ for the full regression
function parameter. As in the previous section
we take the $\eps_i$ to be i.i.d.~$\N(0,\sigma^2)$,
and view the $x_i$ to be coming from a covariate 
distribution with density $f$ with cumulative $F$. 
In this section we work under the model conditions 
given above, where there are underlying true parameters 
$\theta_0=(\gamma_0,a_0,b_0)$ and $\sigma_0$, with 
$\gamma_0$ an inner point of $(0,1)$ and $a_0\not=b_0$.
To develop model robust inference methods and 
for constructing model selection criteria it is 
however also necessary to examine properties of 
estimators outside these model conditions, i.e.~when
the regression function $m_\true(x)$ is not necessarily 
of the precise form implicit in (\ref{eq:twowindowmodel}).
Such an extended analysis is given in 
Section \ref{section:twowindowsB} below.


As for the case examined in Section \ref{section:smooth},
the log-likelihood function is 
$\ell_n(\theta,\sigma)=-n\log\sigma-Q_n(\theta)/\sigma^2$,
so the ML estimators 
$\hatt\theta=(\hatt\gamma,\hatt a,\hatt b)$ 
for are those that minimise
\beq
\label{eq:Qnfortwo}
Q_n(\theta)=\sumin\{y_i-m(x_i,\theta)\}^2
   =\sum_{x_i\le\gamma}(y_i-a)^2+\sum_{x_i>\gamma}(y_i-b)^2, 
\eeq 
i.e.~identical to the least squares estimators.
The numerical recipe is clear enough, since 
$$\hatt a(\gamma)=\bar y_L={\sum_{x_i\le\gamma}y_i
   \over \sum_{x_i\le\gamma}1}
  \quadandquad
  \hatt b(\gamma)=\bar y_R={\sum_{x_i>\gamma}y_i
   \over \sum_{x_i>\gamma}1}, $$
the averages of $y_i$ values to the left and to the right, 
respectively, minimise the function for given $\gamma$. 
This reduces the problem to the computation of 
a one-dimensional curve 
$$Q_n(\gamma,\hatt a(\gamma),\hatt b(\gamma))
   =\sum_{x_i\le\gamma}\{y_i-\hatt a(\gamma)\}^2
   +\sum_{x_i>\gamma}\{y_i-\hatt b(\gamma)\}^2
   =V_{n,L}(\gamma)+V_{n,R}(\gamma), $$
say, splitting into a sum of the variation to the left 
and to the right of $\gamma$. 
This function is constant inside intervals spanned
by consecutive $x_i$ values (i.e.~inside each $(x_{i-1},x_i)$,
if the $x_i$ are ordered), and we define $\hatt\gamma$ 
to be the midpoint inside the relevant interval where
$Q_n$ is smallest. The sizes of these intervals tend
uniformly to zero as sample size increases,
as a consequence of Assumption \ref{assumption:xxiid}, 
so in the limit there is no ambiguity. 
We also note that an algebraic reformulation shows that 
minimising $V_{n,L}(\gamma)+V_{n,R}(\gamma)$ above is 
equivalent to the numerically faster recipe of 
maximising 
$$n_L\hatt a(\gamma)^2+n_R\hatt b(\gamma)^2, $$ 
with $n_L=nF_n(\gamma)$ the number of points to the left 
of $\gamma$ and $n_R=n-n_L$ the number falling to the right. 

Write now 
$$Q_n(\theta)=\sumin\{m(x_i,\theta_0)-m(x_i,\theta)+\eps_i\}^2
   =Q_{n,0}+R_n(\theta)-2V_n(\theta), $$
say, with $Q_{n,0}=\sumin\eps_i^2$ not depending on $\theta$ and with 
\beqn
R_n(\theta)
   =\sumin\{m(x_i,\theta)-m(x_i,\theta_0)\}^2,\quad 
V_n(\theta)
   =\sumin\{m(x_i,\theta)-m(x_i,\theta_0)\}\eps_i. 
\eeqn
It will prove useful to work with the random function
\beq
\label{eq:Hn}
H_n(s,t_1,t_2)=Q_n(\gamma_0+s/n,a_0+t_1/\rootn,b_0+t_2/\rootn)
   -Q_n(\gamma_0,a_0,b_0). 
\eeq 
The reason for the different scales here is that 
it will turn out that $\hatt\gamma$ tends to $\gamma_0$ 
at a faster rate (namely $n$) than do $\hatt a$ and 
$\hatt b$ to $a_0$ and $b_0$ (namely $\rootn$). 
An algebraic reformulation, via $Q_n=Q_{n,0}+R_n-2V_n$ 
above, gives 
$$H_n(s,t_1,t_2)=C_n(s,t_1,t_2)-2D_n(s,t_1,t_2), $$
with 
\begin{align*}
C_n(s,t_1,t_2)
   &=\sumin \{m(x_i,\gamma_0+s/n,a_0+t_1/\rootn,b_0+t_2/\rootn)
   -m(x_i,\gamma_0,a_0,b_0)\}^2, \cr
D_n(s,t_1,t_2)
   &=\sumin \{m(x_i,\gamma_0+s/n,a_0+t_1/\rootn,b_0+t_2/\rootn)
   -m(x_i,\gamma_0,a_0,b_0)\}\eps_i.  
\end{align*}

The following lemma spells out the limiting behaviour of 
$C_n$ and $D_n$, and hence of $H_n$. For this we first need
to describe a certain two-sided compound Poisson process.
For an intensity parameter $\lambda$, let 
$N^*(\lambda,s)$ be a two-sided Poisson process, 
i.e.~equal to a Poisson process $N(s)$ for $s\ge0$
and to another and independent Poisson process $N'(|s|)$
for $s\le0$, both with parameter $\lambda$. 
Next, for $s\ge0$ define 
$$W^*(\lambda,s)=\sum_{i=1}^{N(s)}U_i, $$
with the $U_i$ independent and standard normal
and independent of $N_1$, and similarly 
$W^*(\lambda,s)=\sum_{i=1}^{N'(|s|)}U_i'$
for $s\le0$, with this second set of $U_i'$
also being independent and standard normal
and independent of variables employed for positive $s$. 
We note that $W^*(\lambda,s)$ has mean zero, 
variance $\lambda|s|$, a point mass at zero of 
size $\exp(-\lambda |s|)$, and a somewhat heavier kurtosis
than a two-sided Wiener process with the same variance; 
$$\E \{W^*(\lambda,s)/\sqrt{\lambda|s|}\}^4=3+3/(\lambda|s|). $$

The point is now that $N^*(\lambda,s)$ and $W^*(\lambda,s)$ 
as described here appear as limits of natural 
processes inside our framework, and with $\lambda=f(\gamma_0)$.
First let $s>0$ and consider $N_n(s)$, the number of $x_i$ 
in $[\gamma_0,\gamma_0+s/n]$. It is a binomial $(n,q_n(s))$,
say, with $q_n(s)=F(\gamma_0,\gamma_0+s/n)$. We see
that $N_n(s)\arr_d N(s)$, a Poisson process with intensity
$f(\gamma_0)$, and similarly for $s$ negative. 
Furthermore, the process 
\begin{equation*}
W_n(s)=
\begin{cases}
\sum_{\gamma_0<x_i\le\gamma_0+s/n}\eps_i &\mbox{if\ } s>0, \\
\sum_{\gamma_0+s/n\le x_i<\gamma_0}\eps_i &\mbox{if\ } s<0, 
\end{cases} 
\end{equation*}
converges in distribution to a scaled version 
of such a two-sided compound Poisson process,
say $\sigma_0\,W^*(f(\gamma_0),s)$. 
Proving this is not difficult, e.g.~via 
moment-generating functions. 

\begin{lemma}
\label{lemma:HntoHundermodel}
There is process convergence
$$H_n(s,t_1,t_2)\arr_d H(s,t_1,t_2)=C(s,t_1,t_2)-2D(s,t_1,t_2), $$
where 
\begin{align*}
C(s,t_1,t_2)&=(b_0-a_0)^2N^*(f(\gamma_0),s)
   +F(\gamma_0)t_1^2+\{1-F(\gamma_0)\}t_2^2, \cr
D(s,t_1,t_2)&=\sigma_0|b_0-a_0|W^*(f(\gamma_0),s)
   +V_1t_1+V_2t_2, 
\end{align*}
where $V_1$ and $V_2$ are independent zero-mean normals
with variances $\sigma_0^2F(\gamma_0)$ and 
$\sigma_0^2\{1-F(\gamma_0)\}$ and also independent of 
$N^*(f(\gamma_0),s)$ and $W^*(f(\gamma_0),s)$, introduced above. 
The convergence takes place in the Skorokhod function space
$D([-r,r]^3)$, for each positive~$r$.
\end{lemma}

Note that we only sketch the proof, and do not provide a formal derivation of the tightness required for full process convergence.

\begin{proof}[Sketch of proof]
We first demonstrate that $C_n$ tends to the $C$ function.
To this end, consider the $s>0$ case first. Then 
\begin{align*}
m(x_i,\gamma_0+s/n,a_0&+t_1/\rootn,b_0+t_2/\rootn)
    -m(x_i,\gamma_0,a_0,b_0) \cr
   &=\left\{ \begin{matrix} a_0+t_1/\rootn-a_0 
      &\text{if } x_i\le\gamma_0, \cr
  a_0+t_1/\rootn-b_0 &\text{if } \gamma_0<x_i\le\gamma_0+s/n, \cr
  b_0+t_2/\rootn-b_0 & \text{if } x_i>\gamma_0+s/n. \end{matrix} \right.
\end{align*}
It follows, using $F_n$ for the empirical cumulative distribution
function of the $x_i$ points, and recognising 
$N_n(s)=nF_n(\gamma_0,\gamma_0+s/n)$ from above, that 
\begin{align*}
C_n(s,t_1,t_2)
&=nF_n(\gamma_0)t_1^2/n
   +(a_0-b_0-t_1/\rootn)^2nF_n(\gamma_0,\gamma_0+s/n) \cr
&\qquad\qquad
   +n\{1-F_n(\gamma_0+s/n)\}t_2^2/n \cr
&\arr_p
   (b_0-a_0)^2N^*(\lambda,s)
   +F(\gamma_0)t_1^2+\{1-F(\gamma_0)\}t_2^2,
\end{align*}
writing $\lambda=f(\gamma_0)$. 
Similarly going through the case of $s<0$ gives the 
$C_n\arr_d C$ result, for all $(s,t_1,t_2)$. 

Next we work with $D_n$, which for the case of $s>0$ may
be expressed as 
$$D_n(s,t_1,t_2)=t_1n^{-1/2}\sum_{x_i\le\gamma_0}\eps_i
   +t_2n^{-1/2}\sum_{x_i>\gamma_0+s/n}\eps_i
   +(a_0-b_0+t_2/\rootn)\sum_{\gamma_0<x_i\le\gamma_0+s/n}\eps_i. $$
These three terms are independent, and are seen via
the definition and results noted above to tend to respectively 
$t_1V_1$, $t_2V_2$, $\sigma|b_0-a_0|\,W^*(\lambda,s)$, 
as announced. The case of $s<0$ is handled similarly.

\end{proof}

Because the criterion function is a step function in $s$, 
the standard argmax-argument used e.g.~in \citet{vdVaartWellner96} 
and \citet{Kosorok08} is not valid. This is because the argmax has 
a continuum of solutions.
A way to solve this problem is to we impose uniqueness 
by working with a solution having 
a certain given property, such that it is the smallest 
or largest solution in terms of the value of $s$.
These functionals are called sargmax and
largmax respectively. Other choices can be the middle point, etc.
Asymptotics for the sargmax and largmax is formalised in \citet{SeijoSen11b}. 
Their general theory provides conditions for extending the standard argmax
continuous mapping result to the case when one of the coordinates in
the criterion function is a step-function. This fits in perfectly with
the current setting, and indeed the present section's conclusions are also
reached in \citet{SeijoSen11b}. Their paper does not deal
with the case when several parameters in the criterion function are 
step functions which we will need later on. It seems reasonable that
their results for sargmax and largmax extends to this case as well,
but we do not provide such an extension here. However, we will
from now on understand argmax as sargmax for the coordinates with
non-unique solutions without further comment.

For the current one jump case, the theory in cf.~\citet{SeijoSen11b} 
implies both that $n (\hatt\gamma_n-\gamma_0)=O_\pr(1)$ and
$$\bigl(n(\hatt\gamma-\gamma_0),
\rootn(\hatt a-a_0),\rootn(\hatt b-b_0)\bigr) 
= \argmin(H_n)\arr_d\argmin(H)=(\hatt s,\hatt t_1,\hatt t_2), $$
say. We observe that both $C$ and $D$ above split nicely into
three different terms, involving $s,t_1,t_2$ separately. 
This makes the argmin principle work quite easily, as we may
work out $\hatt s,\hatt t_1,\hatt t_2$ separately,
along with studying their properties.
We find that $\hatt a,\hatt b,\hatt\gamma$ 
are asymptotically independent, with 
\beq
\label{eq:threelimits}
\begin{array}{rcl}
\rootn(\hatt a-a_0)
   &\arr_d& V_1/F(\gamma_0)\sim\N(0,\sigma_0^2F(\gamma_0)^{-1}), \cr
\rootn(\hatt b-b_0)
   &\arr_d& V_2/\{1-F(\gamma_0)\}\sim
   \N(0,\sigma_0^2\{1-F(\gamma_0)\}^{-1}), \cr
n(\hatt\gamma-\gamma_0)
   &\arr_d&\hatt s=\argmax(M). 
\end{array}
\eeq 
Here 
\beq
\label{eq:Msfirsttime}
\begin{array}{rcl}
M(s)&=&\sigma_0|b_0-a_0|W^*(\lambda,s)-\half(b_0-a_0)^2N^*(\lambda,s) \\
&=&\sigma|b_0-a_0|\,\{W^*(\lambda,s)-e_0N^*(\lambda,s)\},
\end{array}
\eeq
where $e_0=\half|b_0-a_0|/\sigma_0$. 
For given intensity $\lambda=f(\gamma_0)$,
the distribution of $\argmax(M)$ depends on the 
$a_0,b_0,\sigma_0$ parameters via $e_0$. 
Ten simulated $M(s)$ curves are displayed 
in Figure \ref{figure:figureC}. 
We have $W^*(\lambda,s)=O_\pr(s^{1/2})$ 
whereas $N^*(\lambda,s)/(\lambda|s|)\arr_\pr1$ 
for growing $|s|$, so $M(s)$ may for large $|s|$
be viewed as roughly equal to $-\half(b_0-a_0)^2\lambda|s|$ 
plus noise of order $\sqrt{|s|}$; in particular, there
is a unique interval $[s_1,s_2]$ at which $M$ 
attains its maximum value. As we let the argmax be the
sargmax, we have that the argmax is defined as $s_1$.

\begin{figure}
\centering
\includegraphics[width = \textwidth]{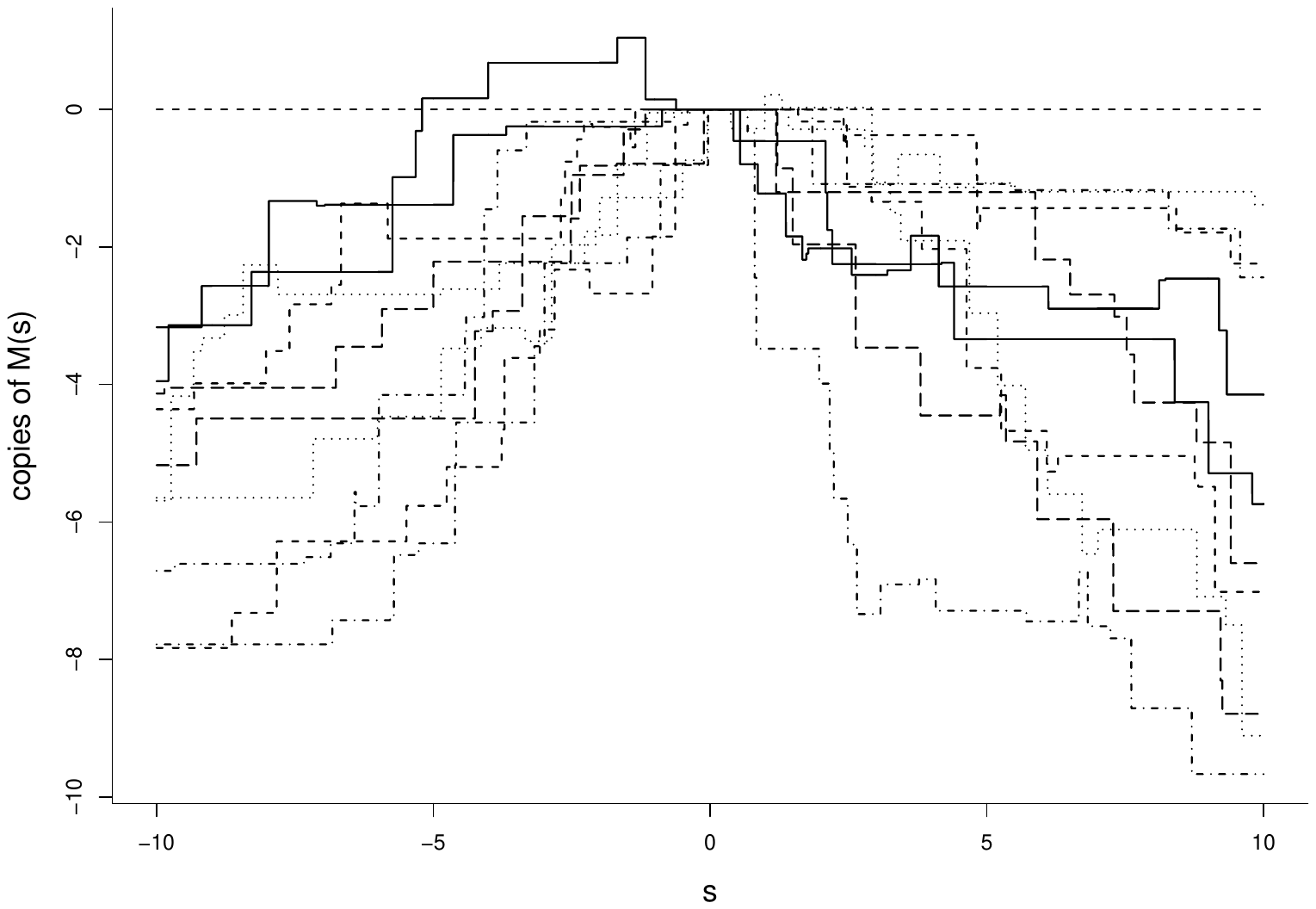}
\caption{
Ten simulated $M(s)$ curves of the type 
(\ref{eq:Msfirsttime}), 
with values $(0.5,2.0,3.0,1.0)$ for $(\sigma_0,a_0,b_0,\lambda)$. 
Here $\hatt s=\argmax(M)$ is the limit distribution of 
$n(\hatt\gamma-\gamma_0)$.} 
\label{figure:figureC}
\end{figure} 

\smallskip
\begin{remark}
\label{remark:remarkA}
{{\rm 
This result makes it possible to construct a confidence 
interval also for the break point parameter $\gamma_0$. 
For given $c$ and $\lambda$ positive, 
let $q(0.975,e_0,\lambda)$ denote the 0.975 quantile 
of the distribution of say 
$G(e_0,\lambda)=\argmax\{W^*(\lambda,s)-\half e_0N^*(\lambda,s)\}$, 
so that $[-q(0.975,e_0,\lambda),q(0.975,e_0,\lambda)]$ 
covers precisely 0.95 of this distribution (since it 
is symmetric). Since $q(0.975,e_0,\lambda)$ is seen 
to be continuous in $(e_0,\lambda)$, 
$$\Pr\{\gamma_0\in\hatt\gamma\pm 
   q(0.975,\hatt e_0,\hatt\lambda)/n\}\arr0.95 $$
as long as $\hatt e_0$ and $\hatt\lambda$ are consistent 
estimators of $e_0=\half|b_0-a_0|/\sigma_0$ 
and $\lambda=f(\gamma_0)$ respectively.  
Here we may simply insert sample-based versions of 
the unknown quantities, including say $\hatt f(\hatt\gamma)$
with a kernel estimator for $f$. 
In some cases the design density is even known 
from the context, e.g.~when it is uniform.}}
\end{remark}


\section{The two-window case, outside model conditions}
\label{section:twowindowsB} 

Extending Lemma \ref{lemma:HntoHundermodel}
and its corollaries to the case 
where the real regression mechanism lies outside 
the (\ref{eq:intro3}) model is of interest 
from several perspectives, and will in particular
be necessary for establishing our jump information
criterion AJIC in Section \ref{section:ajic}. 
We follow and modify the development for the 
previous subsection, and start from the assumption 
that $y_i=m(x_i)+\eps_i$ for i.i.d.~zero-mean errors $\eps_i$, 
for some underlying true regression function $m_\true=m$
that is not necessarily flat on $[0,\gamma_0]$ 
and $(\gamma_0,1]$ for some $\gamma_0$. 
We need to define the least false parameters
$(\gamma_{0,n},a_{0,n},b_{0,n})$ that minimise 
$$\sumin\{m(x_i)-m(x,\gamma_0,a,b)\}^2
   =\sum_{x_i\le\gamma_0}\{m(x_i)-a\}^2
   +\sum_{x_i>\gamma_0}\{m(x_i)-b\}^2. $$
We find that 
$$a_{0,n}=\bar m_L={\sum_{x_i\le\gamma_{0,n}} m(x_i)\over \sum_{x_i\le\gamma_{0,n}} 1}
   \quadandquad 
  b_{0,n}=\bar m_R={\sum_{x_i>\gamma_{0,n}} m(x_i)\over \sum_{x_i>\gamma_{0,n}} 1}, 
   \eqno(2.3)$$
and there will be a certain interval $(x_i,x_{i+1})$ 
of values of $\gamma$ for which the profile function 
$$R_n(\gamma)=\sum_{x_i\le\gamma}\{m(x_i)-\bar m_L\}^2
   +\sum_{x_i>\gamma} \{m(x_i)-\bar m_R\}^2 $$
is smallest. We define $\gamma_{0,n}$ to be the mid-point
of that interval. These least false parameters,
for given $x_1,\ldots,x_n$, have limits  
$$a_0={\int_0^{\gamma_0}m(x)f(x)\,\dd x\over \int_0^{\gamma_0}f(x)\,\dd x}
   \quadandquad
  b_0={\int_{\gamma_0}^1m(x)f(x)\,\dd x\over \int_{\gamma_0}^1f(x)\,\dd x} $$
as $n$ grows to infinity, where we assume there 
is a unique $\gamma_0$ minimising 
$$R(\gamma)=\int_0^\gamma \{m(x)-a_0(\gamma)\}^2f(x)\,\dd x
   +\int_\gamma^1 \{m(x)-b_0(\gamma)\}^2f(x)\,\dd x. $$

Parallelling the development of the previous section, 
we may write
\begin{align*}
H_n(s,t_1,t_2)
   &=Q_n(\gamma_0+s/n,a_{0,n}+t_1/\rootn,b_{0,n}+t_2/\rootn)
   -Q_n(\gamma_0,a_{0,n},b_{0,n}) \cr
   &=C_n(s,t_1,t_2)-2D_n(s,t_1,t_2),
\end{align*}
with 
\begin{align*}
C_n(s,t_1,t_2)
   &=\sumin\bigl[\{m(x_i,\gamma_{0,n}+s/n,a_{0,n}+t_1/\rootn,
   b_{0,n}+t_2/\rootn)-m(x_i)\}^2 \cr
   &\qquad-\{m(x_i,\gamma_{0,n},a_{0,n},b_{0,n})-m(x_i)\}^2\bigr], \cr
D_n(s,t_1,t_2)
   &=\sumin\{m(x_i,\gamma_{0,n}+s/n,a_{0,n}
   +t_1/\rootn,b_{0,n}+t_2/\rootn) \cr
   &\qquad -m(x_i,\gamma_{0,n},a_{0,n},b_{0,n})\}\eps_i.
\end{align*}

For the following result we require that the true
regression function $m$ has left- and right-hand limits
$m(\gamma_0-)$ and $m(\gamma_0+)$ at a genuine 
jump point $\gamma_0$, and furthermore that 
$m(\gamma_0+)$ is closer to $b_0$ than $a_0$
and that $m(\gamma_0-)$ is closer to $a_0$ than $b_0$. 
In particular, the quantities
\beq
\label{eq:c1c2}
\begin{array}{rcl}
c_1&=\{m(\gamma_0-)-b_0\}^2-\{m(\gamma_0-)-a_0\}^2, \\
c_2&=\{m(\gamma_0+)-a_0\}^2-\{m(\gamma_0+)-b_0\}^2 
\end{array}
\eeq
are then both positive. 
Under model conditions, $m(\gamma_0+)=b_0$ and $m(\gamma_0-)=a_0$, 
making both $c_1$ and $c_2$ equal to $|b_0-a_0|^2$,
which is the situation of Lemma \ref{lemma:HntoHundermodel}.
The following generalises that result. 

\begin{lemma}
\label{lemma:HntoHoutside}
Under the above conditions, there is process convergence
$$H_n(s,t_1,t_2)\arr_d H(s,t_1,t_2)=C(s,t_1,t_2)-2D(s,t_1,t_2), $$
where 
\begin{align*}
C(s,t_1,t_2)&=F(\gamma_0)t_1^2+\{1-F(\gamma_0)\}t_2^2
   +\left\{ 
\begin{matrix} 
c_1 N^*(f(\gamma_0),s) & \text{if }s<0, \cr
c_2 N^*(f(\gamma_0),s) & \text{if }s\ge 0, 
\end{matrix} \right. \cr
D(s,t_1,t_2)&=V_1t_1+V_2t_2+\sigma_\true|b_0-a_0|\,W^*(f(\gamma_0),s), 
\end{align*}
where $V_1$, $V_2$, $(N^*,W^*)$ are independent,
with $V_1$ and $V_2$ zero-mean normals with variances
$\sigma_\true^2F(\gamma_0)$ and $\sigma_\true^2\{1-F(\gamma_0)\}$. 
The convergence takes place in the Skorokhod function space
$D([-r,r]^3)$, for each positive $r$.
\end{lemma}

Again, we only provide a sketch of the proof, as we do not provide a formal tightness proof required for full process convergence.
\begin{proof}[Sketch of proof]
That $D_n$ tends to the same type of limit process $D$ as 
for Lemma \ref{lemma:HntoHundermodel}, 
with $\sigma_\true$ now replacing $\sigma_0$ there, 
follows essentially as in the proof of that lemma, 
with some more attention to details. 
Now $C_n$ has a more complicated limit function $C$, however. 

To study this limit, let first $s>0$. The `left' part 
of $C_n$ takes the form  
$$\sum_{x_i\le \gamma_0}[\{a_{0,n}+t_1/\rootn-m(x_i)\}^2
   -\{a_{0,n}-m(x_i)\}^2], $$
 is equal to $F_n(\gamma_0)t_1^2$,
converging again to $F(\gamma_0)t_1^2$. 
The `right' part, with $x_i>\gamma_0+s/n$ terms, 
is a bit more complicated, but ends up converging to 
$\{1-F(\gamma_0)\}t_2^2$. The middle part of $C_n$, 
summing over $x_i$ in $(\gamma_0,\gamma_0+s/n]$, of which there
are $N_n(s)$ terms converging to $N^*(f(\gamma_0),s)$, 
consists of terms essentially equal to 
$$\{m(\gamma_{0,n}+)-a_{0,n}\}^2-\{m(\gamma_{0,n}+)-b_{0,n}\}^2. $$
By assumptions, therefore, this middle part of $C_n(s,t_1,t_2)$
converges to the announced $c_2N^*(f(\gamma_0),s)$ 
when $s>0$. The case of $s<0$ is handled similarly, 
and one finds that the middle part of $C_n$ tends
to $c_1N^*(f(\gamma_0),s)$. 

\end{proof}

A fuller analysis taking also $\sigma$ estimation
into account starts with the log-likelihood function
$\ell_n(\theta,\sigma)
   =-n\log\sigma-\half(1/\sigma^2)Q_n(\theta)$,
as in Section \ref{section:smooth}, but now with 
the more complicated and non-smooth $Q_n(\theta)$
of (\ref{eq:Qnfortwo}). 
The parallel to the crucial $A_n$ function (\ref{eq:An}) 
studied earlier is now 
\beqn
A_n(s,t_1,t_2,u)
=\ell_n(\gamma_0+s/n,a_0+t_1/\rootn,b_0+t_2/\rootn,\sigma_0+u/\rootn)
   -\ell_n(\gamma_0,a_0,b_0,\sigma_0).
\eeqn
Via efforts above this random function 
is seen to converge in distribution to 
\beq
\label{eq:Aforfour}
\begin{array}{rcl}
A(s,t_1,t_2,u)
&=&{u\over \sigma_0}W-{u^2\over \sigma_0^2}
   -\half{1\over \sigma_0^2}H(s,t_1,t_2) \\
&=&{u\over \sigma_0}W-{u^2\over \sigma_0^2}
   +{1\over \sigma_0²} \{D(s,t_1,t_2)-\half C(s,t_1,t_2)\} \\
&=&{u\over \sigma_0}W-{u^2\over \sigma_0^2}
   +{1\over \sigma_0^2} 
   \{V_1t_1+V_2t_2-\half F_1t_1^2-\half F_2t_2^2+M(s)\},
\end{array}
\eeq
say. Here we write for simplicity 
$F_1=F(0,\gamma_0)$ and $F_2=F(\gamma_0,1)$, and
\beq
\label{eq:Mstwosided}
M(s)=\sigma_\true|b_0-a_0|W^*(f(\gamma_0),s)
   -\half c(s)N^*(f(\gamma_0),s), 
\eeq 
with $c(s)$ being $c_1$ for $s$ positive 
and $c_2$ for $s$ negative. This properly generalises
the case of (\ref{eq:Msfirsttime}), where
$c_1$ and $c_2$ of (\ref{eq:c1c2}) both are equal 
to $(b_0-a_0)^2$.

We learn from this that $\hatt a,\hatt b,\hatt\gamma,\hatt\sigma$
are again asymptotically independent, 
with $\hatt a$ and $\hatt b$ having 
the same limit distributions as in (\ref{eq:threelimits}), 
but that the limit distribution for $\hatt\gamma$
is now more complicated; 
$n(\hatt\gamma-\gamma_0)\arr_d\argmax(M)$.  We note that a fully rigorous proof of this result requires the simultaneous process convergence of $A$ and $A$'s pure jump process, as explained in \citet{SeijoSen11b}.

\section{The multi-windows model} 
\label{section:multiwindows}

We now proceed to the more general case of several break points
and windows. The theory and results of the previous section
generalise rather nicely, without severe complications.

Consider therefore model (\ref{eq:intro2}) with 
the regression function is flat inside each of $d$ windows. 
The estimators $\hatt a_j$ and $\hatt\gamma_j$ 
are again those that minimise
$Q_n(\theta)=\sumin\{y_i-m(x_i,\theta)\}^2$. 
This is minimised for given $\gamma$ by setting
$a_j$ equal to $\hatt a_j(\gamma)$, the average of $y_i$ 
in window $G_j(\gamma)=(\gamma_{j-1},\gamma_j]$. Hence 
windows are chosen such that the combined variation 
$$Q_{n,\prof}(\gamma)
   =Q_n(\hatt a(\gamma),\gamma)
   =\sum_{j=1}^d\sum_{i\in G_j(\gamma)}\{y_i-\hatt a_j(\gamma)\}^2 $$
is minimal. This is equivalent to maximising 
$\sum_{j=1}^d n_j\hatt a_j(\gamma)^2$.

In this section we do assume that the (\ref{eq:intro2}) 
model is fully correct, with true parameter values 
$a_{j,0}$ (there are $d$ such) 
and $\gamma_{j,0}$ (there are $d-1$ such). Define
$$H_n(s,t)=Q_n(\gamma_0+s/n,a_0+t/\rootn)-Q_n(\gamma_0,a_0)
   =C_n(s,t)-2D_n(s,t), $$
where $t=(t_1,\ldots,t_d)^\tr$ and $s=(s_1,\ldots,s_{d-1})^\tr$. 
The methods and results of 
Sections \ref{section:twowindowsA}--\ref{section:twowindowsB} 
are not very hard to generalise, and lead to 
$$H_n(s,t)\arr_d H(s,t)=C(s,t)-2D(s,t), $$
with 
\begin{align*}
C(s,t)&=\sum_{j=1}^{d-1} c_j(s) N_j^*(\lambda_j,s_j)
   +\sum_{j=1}^d F(\gamma_{j-1,0},\gamma_{j,0})t_j^2, \cr
D(s,t)&=\sum_{j=1}^{d-1}\sigma_\true |a_{j+1,0}-a_{j,0}|
   W^*_j(\lambda_j,s) + \sum_{j=1}^d V_jt_j, 
\end{align*}
where $W_1^*(\lambda_1,s),\ldots,W_{d-1}^*(\lambda_{d-1},s)$ 
are independent two-sided compound Poisson processes, 
as described in association with Lemma \ref{lemma:HntoHundermodel},
with $\lambda_j=f(\gamma_{0,j})$, 
and independent of the $V_j$,
which are independent zero-mean normals with variances
$\sigma^2F(\gamma_{0,j-1},\gamma_{0,j})$. 
Here we write $F(c,d)$ for $F(d)-F(c)$. 
Finally, $c_j(s)$ above is $c_{j,1}$ for $s>0$ and $c_{j,2}$ 
for $s<0$, where 
\beq
\label{eq:cj1cj2}
\begin{array}{rcl}
c_{j,1}&=&\{m(\gamma_{0,j}-)-a_{0,j+1}\}^2
   -\{m(\gamma_{0,j}-)-a_{0,j}\}^2, \\ 
c_{j,2}&=&\{m(\gamma_{0,j}+)-a_{0,j}\}^2
   -\{m(\gamma_{0,j}+)-a_{0,j+1}\}^2. 
\end{array}
\eeq 
When the model used is correct, then $m(\gamma_{0,j}-)=a_{0,j}$
and $m(\gamma_{0,j}+)=a_{0,j+1}$, so that both $c_{j,1}$ 
and $c_{j,2}$ are equal to $(a_{0,j+1}-a_{0,j})^2$,
cf.~Lemma \ref{lemma:HntoHundermodel}.

The crucial corollary is that all $2d$ 
parameter estimators are asymptotically independent, with 
\begin{align*}
\rootn(\hatt a_j-a_{j,0})
   &\arr_dV_j/F(\gamma_{j-1,0},\gamma_{j,0})
   \sim\N(0,\sigma^2/F(\gamma_{j-1,0},\gamma_{j,0})), \cr
n(\hatt\gamma_j-\gamma_{j,0})
   &\arr_d \argmax(M_j), 
\end{align*}
where $c_j=\half|a_{j+1,0}-a_{j,0}|/\sigma$
and $\lambda_j=f(\gamma_{0,j})$, featuring also 
\beq
\label{eq:Mjs}
M_j(s)=\sigma_\true |a_{0,j+1}-a_{0,j}| W_j^*(\lambda_j,s)
   -\half c_j(s) N_n^*(\lambda_j,s). 
\eeq 
The $c_j,\lambda_j$ quantities may be estimated 
consistently, after which confidence intervals 
may be computed for each break point $\gamma_j$, 
as per Remark \ref{remark:remarkA}.

Note that to fully formalise this argument, an extension of 
the theory in \citet{SeijoSen11b} to the case with several
piecewise constant coordinates are needed. We do not
provide this extension here.

\section{The AJIC} 
\label{section:ajic}

To develop a proper analogue of the AIC and its 
model robust version $\AIC^*$ we need to follow 
the chain of arguments given in Section \ref{section:smooth},
but to modify results and formulae under way
to take the nonstandard asymptotics into account. 
This is what we focus on in this section, aiming
for a formula of the type $\AIC^*=2(\ell_{n,\max}-\hatt b)$,
with $\hatt b$ the estimate of the positive bias $b$
involved when assessing $\ell_{n,\max}/n$ as an estimator
of the model specific part of the Kullback--Leibler
divergence from truth to model. 

We begin with 
\beqn
A_n(s,t,u)&=&
   \ell_n(\gamma_{0,n}+s/n,a_{0,n}+t/\rootn,\sigma_{0,n}+u/\rootn)
   -\ell_n(\gamma_{0,n},a_{0,n},\sigma_{0,n}) \\
   & &\arr_d A(s,t,u)
   ={u\over \sigma_0}W-{u^2\over \sigma_0^2}
   -\half{1\over \sigma_0^2}\{C(s,t)-2D(s,t)\}.
\eeqn 
The limit process may be written in the form 
\beqn
A(s,t,u)={u\over \sigma_0}W-{u^2\over \sigma_0^2}
   +{1\over \sigma_0^2}\sum_{j=1}^d 
   \{V_jt_j-\half F(\gamma_{0,j-1},\gamma_{0, j})t_j^2\}
   +{1\over \sigma_0^2}\sum_{j=1}^{d-1} M_j(s_j),
\eeqn 
with $M_j(s)$ as in (\ref{eq:Mjs}). 
Here we need 
$$A_n(\hatt s_n,\hatt t_n,\hatt u_n)
   =\ell_n(\hatt\theta,\hatt\sigma)-\ell_n(\theta_{0,n},\sigma_{0,n})
   \arr_d A(\hatt s,\hatt t,\hatt u), $$
where $\hatt t_j=V_j/F(\gamma_{0,j-1},\gamma_{0,j})$ 
and $\hatt u=\half\sigma_0W$, along with 
the more involved $\hatt s_j=\argmin(H_j)=\argmax(M_j)$. 
So $\ell_{n,\max}=\ell_{n,0}+E_{n,1}$
with 
$$E_{n,1}\arr_d A(\hatt s,\hatt t,\hatt u)
   =\quart W^2+\half{1\over \sigma_0^2}
   \sum_{j=1}^d {V_j^2\over F(\gamma_{0,j-1},\gamma_{0,j})}
   +{1\over \sigma_0^2}\sum_{j=1}^{d-1} M_j(\hatt s_j).$$

Next up is the parallel to the $B_n$ function 
of Section \ref{section:smooth}, which takes the form
$$B_n(s,t,u)=k_n(\gamma_{0,n}+s/n,a_{0,n}+t/\rootn,\sigma_{0,n}+u/\rootn)
  -k_n(\gamma_{0,n},a_{0,n},\sigma_{0,n}). $$
We find that it tends in distribution to 
$$B(s,t,u)=-{u^2\over \sigma_0^2}-\half{1\over \sigma_0^2}C(s,t) $$
with 
$$C(s,t)=\sum_{j=1}^{d-1} e_j(s)N_j^*(\lambda_j,s)
    +\sum_{j=1}^d F(\gamma_{j-1,0},\gamma_{j,0})t_j^2, $$
where $e_j(s)$ is $d_{1,j}$ for $s$ positive and $d_{2,j}$
for $s$ negative. With 
$$k_n(\hatt\theta,\hatt\sigma)
   =k_n(\theta_{0,n},\sigma_{0,n})+E_{n,2}, $$
therefore, we have
$$E_{n,2}=B_n(\hatt s_n,\hatt t_n,\hatt u_n)
   \arr_d E_2=B(\hatt s,\hatt t,\hatt u), $$
yielding
$$E_2=-\quart W^2-\half{1\over \sigma_0^2}
   \sum_{j=1}^d {V_j^2\over F(\gamma_{0,j-1},\gamma_{0,j})}
   -\half{1\over \sigma_0^2}
   \sum_{j=1}^{d-1} c_j(\hatt s_j)N_j^*(\lambda_j,s). $$
This means that the extra part, the positive 
random component of $\ell_{n,\max}$ as an estimator of 
$k_n(\hatt\theta,\hatt\sigma)$, tends in distribution to 
$$E_1-E_2=\half W^2+{1\over \sigma_0^2}
   \sum_{j=1}^d{V_j^2\over F(\gamma_{0,j-1},\gamma_{0,j})}
   +{1\over \sigma_0^2}\sum_{j=1}^{d-1}
   \{M_j(\hatt s_j)+\half c_j(\hatt s_j)N_j^*(\lambda_j,\hatt s_j)\}. $$
The expected bias, for large $n$, is hence 
$$b=1+d{\sigma_\true^2\over \sigma_0^2}
   +{1\over \sigma_0^2}\sum_{j=1}^{d-1}\kappa_j, $$
say, with 
$$\kappa_j
=\E\,\{M_j(\hatt s_j)+
   \half c_j(\hatt s_j)N_j^*(\lambda_j,\hatt s_j)\}
=\sigma_\true|a_{0,j+1}-a_{0,j}|\,\E\,W_j^*(\lambda_j,\hatt s_j). $$
We conclude that the model robust AIC style JIC for
comparing and assessing jump regression models with 
different numbers of windows is 
\beq
\label{eq:ajicstar}
\AJIC^*=2\ell_{n,\max}-2\hatt b
   =2(-n\log\hatt\sigma_0-\half n)
   -2\Bigl(1+d{\hatt\sigma^2\over \hatt\sigma_0^2}
   +{1\over \hatt\sigma_0^2}\sum_{j=1}^{d-1}\hatt\kappa_j\Bigr).
\eeq
This score is computed for each candidate number 
$d$ of windows, say up to a pre-selected $d_{\max}$,
and in the end the model is selected which has
the largest such score. 

As indicated $\ell_{n,\max}$ here is still the simple formula
$-n\log\hatt\sigma_0-\half n$, but with 
$\hatt\sigma_0^2=Q_n(\hatt\gamma,\hatt a)/n$ 
involving the estimation of break points and levels
for the number $d$ of windows under consideration. 
Also, $\hatt\sigma$ is an estimate of $\sigma_\true$,
e.g.~taken as the $\hatt\sigma_0$ for the biggest
$d_{\max}$ under consideration. The most tricky part
of $\AJIC^*$ is the $\hatt\kappa_j$, an estimate
of $\kappa_j$. Since no closed-form expression 
for this expected value is available 
we resort to simulation, taking $\hatt\kappa_j$ 
to be the average of say 1000 simulated copies of 
$\hatt\sigma|\hatt a_{j+1}-\hatt a_j|W^*(\hatt\lambda_j,\hatt s_j)$.
Here $\hatt\lambda_j$ is a consistent estimate of 
$f(\gamma_{0,j})$, e.g.~using a kernel estimate for
the $x_i$; in cases where the design density is known,
however, as when the $x_i$ are seen as uniformly
distributed on the interval in question, we
simply use that known value of $f(\gamma_{0,j})$. 
See Section \ref{section:illustrations} for an illustration.

\section{The BJIC} 
\label{section:bjic}

Here we aim for a selection criterion of the BIC 
type, but with the necessary modifications of
the chain of arguments used to reach 
(\ref{eq:bicforsmooth}) for the smooth regression 
case. As we saw in the BIC subsection of 
Section \ref{section:smooth}, this needs approximations
to marginal likelihoods, and this pans out 
different in jump models as we shall see now. 

In Section \ref{section:ajic} we worked with 
the random process $A_n(s,t,u)$ and its limit $A(s,t,u)$,
along with their maximisers, respectively 
$\hatt s_n,\hatt t_n,\hatt u_n$ and 
$\hatt s,\hatt t,\hatt u$. We infer from these
results that 
\beqn
A_n^*(s,t,u)
&=&A_n(\hatt s_n+s,\hatt t_n+t,\hatt u_n+u)
   -A_n(\hatt s_n,\hatt t_n,\hatt u_n) \\
&=&\ell_n(\hatt\gamma+s/n,\hatt a+t/\rootn,\hatt\sigma+u/\rootn)
   -\ell_n(\hatt\gamma,\hatt a,\hatt\sigma) \\
&\arr_d& A^*(s,t,u)
   =A(\hatt s+s,\hatt t+t,\hatt u+u)-A(\hatt s,\hatt t,\hatt u),
\eeqn
describing the log-likelihood function behaviour
around the ML values. We find 
$$A^*(s,t,u)
   =-{u^2\over \sigma_0^2}-\half{1\over \sigma_0^2}
   \sum_{j=1}^d F_jt_j^2
   +{1\over \sigma_0^2}\sum_{j=1}^{d-1}M_j(s_j), $$
again with $M_j(s_j)$ as in (\ref{eq:Mjs}), and
with $F_j$ short-hand notation for $F(\gamma_{0,j})-F(\gamma_{0,j-1})$. 

We are now in a position to approximate marginal 
likelihoods, which will lead to our BJIC. 
For a given number of windows parameter $d$, 
suppose a prior $\pi(\gamma,a,\sigma)$ is placed 
on the model parameters. 
The marginal likelihood is then 
\beqn
\lambda_n
&=&\int\pi(\gamma,a,\sigma)\exp\{\ell_n(\gamma,a,\sigma)\}
   \,\dd\gamma\,\dd a\,\dd\sigma \\
&=&\exp(\ell_{n,\max})
   \int\pi(\hatt\gamma+s/n,\hatt a+t/\rootn,\hatt\sigma+u/\rootn) \\
& &\qquad\qquad\times
   \exp\{A_n^*(s,t,u)\}\,\dd s\,\dd t\,\dd u\,
   \Big({1\over n}\Bigr)^{d-1}
   \Big({1\over \rootn}\Bigr)^{d}
   \Big({1\over \rootn}\Bigr).
\eeqn 
This leads to the approximation
\beqn
\log\lambda_n
   &=&\ell_{n,\max}-(d-1)\log n-\half(d+1)\log n \\
& &\qquad   
   +\log\int\exp\{A^*(s,t,u)\}\,\dd s\,\dd t\,\dd u
   +\log\pi(\hatt\gamma,\hatt a,\hatt\sigma),
\eeqn
assuming that the prior used is continuous on 
a region containing the parameter estimates. 
We may evaluate the integral part here, in part
using simulations, but it is at any rate clear that 
\beq
\label{eq:loglambdabjic}
\log\lambda_n=\ell_{n,\max}-\half(3d-1)\log n+\delta_n, 
\eeq
say, with $\delta_n$ of lower order than the two leading 
terms. This leads to the modified BIC for jump regression 
models, with 
\beq
\label{eq:bjic}
\BJIC_d=2\ell_{n,d,\max}-(3d-1)\log n, 
\eeq
for the case of model $M_d$, with $d$ windows. Here
$\ell_{n,d,\max}=-n\log\hatt\sigma_d-\half n$,
the log-likelihood maximum for model $M_d$
with $\hatt\sigma_d^2=n^{-1}Q_n(\hatt a)$ the ML
estimate under that model. Increasing $d$
leads decreasing $\hatt\sigma_d$ and hence 
increasing the log-likelihood maximum. 

The traditional BIC suggests penalty $2d\log n$, 
but our version (\ref{eq:bjic}) is the more correct 
version, reflecting the different type of asymptotics. 
As with (\ref{eq:bicforsmooth}), the interpretation
is that when candidate models $M_1,\ldots,M_{d_{\max}}$
are being considered, associated with number of 
windows $d$ being equal to $1,\ldots,d_{\max}$, 
and with prior probabilities
$p(M_1),\ldots,p(M_{d_{\max}})$, then 
$p(M_d\midd\data)\propto p(M_d)\exp(\half\BJIC_d)$. 
With equal prior probabilities for the $d_{\max}$
models, the model is selected with the highest
$\BJIC$ score. 

It is sometimes worth the extra trouble to construct
a more complete approximation to the marginal 
likelihoods, using (\ref{eq:loglambdabjic}) but including
the remainder term $\delta_n$, for each candidate model.
First, 
$$\int\exp\{A^*(s,t,u)\}\,\dd s\,\dd t\,\dd u
   =(2\pi)^{1/2}{\sigma_0\over \sqrt{2}}
   \prod_{j=1}^d(2\pi)^{1/2}{\sigma_0\over F_j^{1/2}}
   \prod_{j=1}^{d-1}\int\exp\Bigl\{{M_j(s_j)\over \sigma_0^2}\Bigr\}\,
   \dd s_j, $$
with logarithm say 
$$r_1=\half(d+1)\log(2\pi)+(d+1)\log\sigma_0-\half\log 2
   -\half\sum_{j=1}^d\log F_j+\sum_{j=1}^{d-1}\log\mu_j, $$
where $\mu_j=\int\exp(M_j/\sigma_0^2)\,\dd s$. 
Secondly, suppose we use independent priors for 
$\gamma$, the levels $a_1,\ldots,a_d$, and $\sigma$,
which appears to be a natural type of construction. 
If we use the natural flat Dirichlet prior
for $\gamma$ on the simplex of $\gamma_j-\gamma_{j-1}$ 
nonnegative and adding to one, and priors for the $a_j$ 
corresponding to uniform distributions
on intervals of length say $\tau$, then
$$r_2=\log\pi(\hatt\gamma,\hatt a,\hatt\sigma)
   =\log\Gamma(d)-d\log\tau+\log\pi_0(\hatt\sigma), $$
say, with $\pi_0(\sigma)$ denoting the prior used
for that parameter. The suggestion is then to 
use the fine-tuned 
$$\BJIC=2\ell_{n,d,\max}-(3d-1)\log n+2(\hatt r_1+r_2), $$
where $\hatt r_1$ inserts estimates for the $F_j$
and uses simulation to evaluate the $\mu_j$, 
and where $r_2$ uses some $\tau$ natural for
the application at hand. Finally we advocate omitting the
$\log\pi_0(\hatt\sigma)$ term in $r_2$ since that term
ought to be approximately the same across models. 

\section{Bayes estimators} 
\label{section:bayes}

Consider the Bayesian framework, which starts out with 
a joint prior density for the level and break point parameters
along with the spread parameter, 
to be combined in the usual fashion with the data likelihood,
which we take to be based on the normal assumption. 
For simplicity of presentation we focus here on the case
of two windows, as in Section \ref{section:twowindowsA}; 
the generalisation to the multi-window case will be fairly 
straightforward, using observations and techniques already 
worked with in Section \ref{section:multiwindows}.

The starting point is a prior density $\pi(\gamma,a,b,\sigma)$ 
for the four unknown parameters;
we would typically take these four to be independent. 
The posterior density is then of course proportional to 
$\pi(\gamma,a,b,\sigma)\exp\{\ell_n(\gamma,a,b,\sigma)\}$, 
with $\ell_n$ the log-likelihood function. 
The random vector
$$(S_n,T_{n,1},T_{n,2},U_n)
   =(n(\gamma-\gamma_0),\rootn(a-a_0),
   \rootn(b-b_0),\rootn(\sigma-\sigma_0)) $$
hence has posterior density, say $\pi_n(s,t_1,t_2,u)$, 
proportional to 
$$\pi(\gamma_0+s/n,a_0+t_1/\rootn,b_0+t_2/\rootn,\sigma_0+u/\rootn)
   \exp\{A_n(s,t_1,t_2,u)\}, $$
where $A_n$ is the log-likelihood process function
worked with in Section \ref{section:twowindowsB},
converging in distribution to the $A(s,t_1,t_2,u)$
of (\ref{eq:Aforfour}). As long as the prior is 
positive and continuous on a region containing
the true values, therefore, the posterior density 
converges almost surely to $c\exp\{A(s,t_1,t_2,u)\}$,
with $c$ the appropriate normalising constant. 
Due to the additive structure of $A$ this means that 
$S_n,T_{n,1},T_{n,2},U_n$ are asymptotically independent
in the Bayesian posterior framework. The most 
interesting aspect here is that the posterior 
density of $S_n=n(\gamma-\gamma_0)$ tends to 
$$h(s)\propto\exp\{M(s)/\sigma_0^2\}, $$
with $M(s)$ as in (\ref{eq:Msfirsttime}). 

Taking the mean here, we also learn that 
$$n(\hatt\gamma_B-\gamma_0)\arr_d K=\int sh(s)\,\dd s
   ={\int s\exp\{M(s)/\sigma_0^2\}\,\dd s \over
     \int \exp\{M(s)/\sigma_0^2\}\,\dd s}, $$
 where $\hatt\gamma_B=\E(\gamma\midd\data)$ is the
Bayes estimator under quadratic loss. 
General results of \citet[Chapter 1.9]{IbragimovKhasminski81}
imply that the Bayes estimator $\hatt\gamma_B$
will have smaller limiting mean squared error than 
that of the ML estimator, i.e.~$K$ above has 
smaller squared error than $\hatt s=\argmax(M)$. 
In particular, there is no \Bernstein--von Mises theorem 
here about the ML and the Bayes estimators doing 
equally well for large $n$. The situation
is easier and more standard when it comes to level
parameters and the $\sigma$, however; here the above
efforts indeed lead to such \Bernstein--von Mises 
theorems. In particular, the posterior distribution of
$\rootn(a-\hatt a)$ has the very same limit as has
$\rootn(\hatt a-a_0)$, etc. 

\section{Algorithms for computing estimates}
\label{section:algorithms}

The standard way of maximising the likelihood of break point 
models is via a dynamic programming approach, 
described e.g.~\citet{BaiPerron98}.
We have followed this approach, and implemented the 
computational algorithm in C. We also improved the 
algorithm in two ways. First, we calculated the likelihood 
contributions of placing a jump at a specific regions 
in the inner-loop of the algorithm. 
Most implementations we have seen calculate these 
contributions in a pre-calculated look-up table. 
Such a table is of size $n\choose d$ and requires 
a vast size if $n$ and $d$ are large. In bioinformatics 
applications, both $n$ and $d$ are routinely very large. 
Another improvement we carried out is that the dynamic programming 
routine calculates the accumulated likelihoods 
of all possible jump-configurations. As the dynamic programming 
routine places one jump at the time, we can avoid calculating 
all further nodes if the current node already have 
an accumulated likelihood exceeding a known likelihood 
configuration. If this likelihood configuration 
is close to the maximised likelihood, we reduce the 
computational burden of the full maximisation. We 
propose the following very fast linear time algorithm 
to provide such an initial estimate.

First, let $i=1$, $u_0=1$, $u_2=n$ and
$$u_1 = \argmin_{u_1} \sum_{j = 0}^{1}\sum_{i=u_j}^{u_{j+1}}
   (Y_i-\bar Y_{u_{j}:u_{j+1}})^2. $$
Note that this is a one-dimensional optimisation. 
Iterate the following two steps until $i=n$:

\begin{enumerate}
\item Let $u_{i+1} = n$ and compute
$$\hatt u_i = \argmin_{u_i}\sum_{j = 0}^{i-1} \sum_{i =  u_j}^{u_{j+1}}
    (Y_i - \bar Y_{ u_{j}:u_{j+1}})^2, $$
and set $u_i = \hat u_i$.
\item
Compute
$$v_i = \argmin_{u_i} \sum_{1 \leq j \leq i , j \neq i-1} 
   \sum_{i =  u_j}^{ u_{j+1}} 
   (Y_i - \bar Y_{u_{j}:u_{j+1}})^2, $$
and set $v_i = \hatt v_i$. If $v_i=u_{i-1}$, increase $i$. 
Now go to step 1. Otherwise, set $u_{i-1}=v_i$ and decrease $i$ 
by 1. Now repeat step 2.
\end{enumerate}
Now let $(w_1, \ldots, w_d)$ be the sorted version of 
$(u_1,\ldots,u_d)$. This is the initial value of the full optimisation.

To test the efficiency of this procedure, we ran a simulation 
with $1\times 10^4$ repetitions, each repetition 
having a random sample size $n$ and a random number 
of break points $d$ chosen uniformly between specified limits. 
The jumps were distributed according to $\N(0,1.5^2)$. 
Table \ref{table::alg} shows the savings when using the above 
initial-value algorithm compared to a full ML 
optimisation. Here, the column with `correct breaks' is the percent 
of jump points shared with the full ML 
solution, `correct solution' is the percentage of times 
where the initial and the ML solutions 
were identical, and the computational reduction 
compares the computation time with or without 
the initial value code. It seems clear that we can gain 
much from specifying a good starting solution 
for the optimisation routine.

\begin{table}[htp]
\begin{center}
\begin{tabular}{lllllll}
  \hline
  Sample & Number of & Correct  & Correct & Computational \\
  size   & breaks    & breaks &   solution & reduction \\
  \hline
  $20 \leq n \leq 2000$ & $1 \leq d \leq \sqrt{n}$  & $54 \%$ & $23 \%$ & $66 \%$ \\
  $20 \leq n \leq 2000$ & $1 \leq d \leq \sqrt{n}/2$  & $71 \%$ & $45 \%$ & $50 \%$ \\
  $20 \leq n \leq 4000$ & $1 \leq d \leq \sqrt{n}/2$  & $65 \%$ & $35 \%$ & $52 \%$ \\
  \hline
\end{tabular}
\end{center}
\caption{The effects of the initial-value computation.}
\label{table::alg}
\end{table}







\section{Illustrations} 
\label{section:illustrations}


The following simple illustration is intended as 
a simple `proof of concept', where 
we demonstrate that the change points models 
and the AJIC machinery beats the smooth regression 
models and the model robust AIC of Section \ref{section:smooth}, 
if there are actual change points in the true model
and these are decently separated. With larger
sample size also windows with smaller differences
in levels may be detected too. 

In Figure \ref{figure:AJIC_and_AIC_scores} 
we have plotted $n = 1000$ realisations from 
a model with three change points 
$\gamma=(0.234,0.50,0.73)$ with levels 
$a=(1.0, 3.1, 2.8, 1.5)$ for the four cells.
The $x_i$ have been sampled from the uniform distribution
on the unit interval and the errors are independent 
Gaussian noise with standard deviation 
$\sigma_\true = 0.5$. We will compare the performance 
of seven models: three change points models with 
1--3 break points and four smooth regression models 
represented by polynomial regression from constant 
to cubic, see Table 
\ref{table:AJIC_and_AIC_scores} for the individual 
results. 

\begin{figure}[h!]
\includegraphics[width = \textwidth]{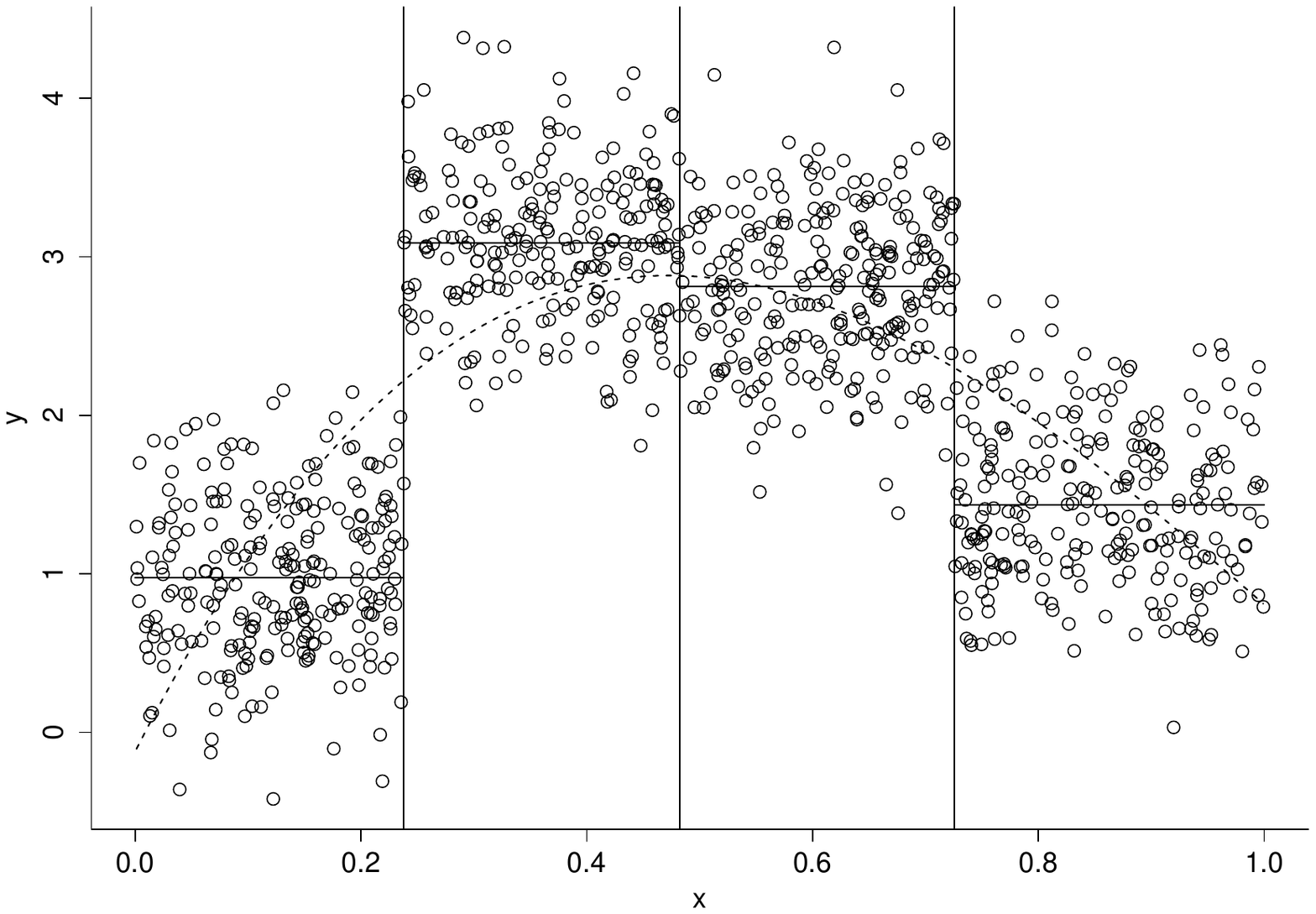}
\caption{
A simulated illustration based on 1000 realisations from 
a true change point model with three break points.
The AJIC winner is the model with three breaks, indicated 
by the vertical lines. The best smooth regression model, 
according to the model robust AIC, is a cubic polynomial 
regression model shown by dashed line above.}
\label{figure:AJIC_and_AIC_scores}
\end{figure}

The algorithm presented in Section \ref{section:algorithms} 
is used to find the break points for the different change 
points models. To estimate the model robust correction 
term $\hatt b$ we need the expectation 
$\E \, W_j^*(\lambda,\hatt s_j)$, 
for $j = 1, \ldots, d - 1$ in each candidate model. 
These expected values do not have closed form expressions. 
A simple practical solution is to evaluating 
these expected values by simulating independent realisations 
of $W_j^*(\lambda,\hatt s_j)$, for each $j$, 
and for each candidate model. For this particular model, 
the estimates seems to be quite stable and 500 realisations 
results in good approximations. Increasing the number of realisations 
in the simulations results in minor improvements, 
i.e.~in the second and third decimal place. 


\begin{table}[h!]
\begin{tabular}{l || r  | r | r r r r r r  r }
  model & AJIC$^{\ast}$ & AIC$^{\ast}$ & AIC & $\ell_{n, \max}$ 
& $\hatt{b}$ & $\hatt{\sigma}_{0}$ & $d$ &  degree & $p$ \\
  \hline                       
  1 & -535.957            & . &  -535.947&-266.062 & 1.917 & 0.791 & 2 & .  & 4\\
  2 &  387.025            & . &   384.540 & 197.883   & 4.371 & 0.497 & 3 & .  & 6\\
  3 &  423.834            & . &   418.901&  217.651  & 5.734 & 0.487 & 4& .  & 8\\
\hline
  constant  & . & -1029.331    &  -1029.887& -513.387  & 1.278 & 1.013 & . & 0 & 2\\
  linear    & . & -1012.365    &  -1012.931& -504.616  & 1.566 & 1.004 & . & 1 & 3\\
  quadratic & . & -293.802     &   -294.966 & -144.155 & 2.746 & 0.701 & . & 2 & 4 \\
  cubic     & . & -270.732     &   -271.925 & -131.981 & 3.385 & 0.692 & . & 3 & 5\\
\end{tabular}
\caption{
Comparison of AJIC and AIC scores for the three change 
points models and four smooth polynomial regression 
models from constant to cubic used in 
Figure \ref{figure:AJIC_and_AIC_scores}. 
The smooth regression models are not even close 
with respect to the change point models, however, 
the difference primarily due to the large gap 
in $\ell_{n, \max}$ values.
}
\label{table:AJIC_and_AIC_scores}
\end{table}

Note that the $\AJIC^*$ and $\AIC^*$ scores, for respectively
the jump models and the smooth models, are constructed
to be directly comparable. In other words, our AJIC machinery
is not merely to compare and rank models with different
number of windows, but to work in concert with the robust
AIC scores for other regression models. 

\begin{remark}
{{\rm 
In the implementation, there is no lower bound 
on the minimum size of the windows. This means 
that if $d_{\max}$  is large, the maximum likelihood 
estimator will typically isolate some single extreme 
points and view these as `thin windows'. 
These observations are potential outliers or 
artifacts of the window search procedure 
and would need to be further examined. For the present 
purpose, however, with focus on model selection, 
a lower bound should be included, since we will 
typically avoid interpreting single points 
or tiny windows with very few points as genuine windows, 
i.e.~not as the effect of real changes in the 
underlying signal. Moreover, the specification 
of a suitable minimum length is also needed 
in order to obtain proper estimates
for $c_{1, j}$ and $c_{2, j}$ in (\ref{eq:cj1cj2}).
This is currently done by using small local averages, which 
implicit assumes a certain minimum window length.
}}
\end{remark}

\section{Concluding remarks}
\label{section:concludingremarks}

We end our paper with a list of concluding notes,
some pointing to further research work of relevance. 

\smallskip

\smallskip
{\sl A. Time series with break points.}
Our basic model has been that of $y_i=m(x_i)+\eps_i$,
with the error terms $\eps_i$ being i.i.d. 
For several application areas these might easily
be dependent, however, and perhaps with the 
$x_i$ signifying either time of spatial context.
One may then extend our methods and results to
such models, e.g.~with and autoregressive structure
for the error terms. Such extensions are possible,
with the required generalisations from the theory
of empirical processes for dependent data. 

\smallskip
{\sl B. Tests for break points.}
Our theory has been geared towards finding good estimates,
along with precision measures and e.g.~confidence
intervals, for given models. The theory may also
be used to develop statistical tests for hypotheses
of the type `$m(x)$ is constant' against 
the alternative that `$m(x)$ has a discontinuity'.
This is not pursued here, however. 

\smallskip
{\sl C. Alternative large-sample setups.}
We have used Assumption \ref{assumption:xxiid}, but
there are other sensible asymptotic approximations as well.
Besides the asymptotic approximations used in the econometric
literature as mentioned in the introduction, there are
settings where deterministic covariates such as $x_i=i/n$ 
are natural. We still have convergence of ergodic averages,
but with different limiting processes.

\bibliographystyle{biometrika}
\bibliography{ghh_bibliography8}


\end{document}